\documentclass[a4paper]{article}
\usepackage{amsfonts}
\usepackage{amsmath}
\usepackage{amsthm}
\usepackage[normalem]{ulem}
\usepackage{xcolor}

\newcommand{\comment}[1]{}

\newcommand{\tr}{\mathrm{tr}}

\newcommand{\Nbb}{\mathbb{N}}
\newcommand{\Rbb}{\mathbb{R}}
\newcommand{\Cbb}{\mathbb{C}}
\newcommand{\Bcal}{\mathcal{B}}
\newcommand{\Ccal}{\mathcal{C}}

\newcommand{\Ical}{\mathcal{I}}
\newcommand{\Lcal}{\mathcal{L}}
\newcommand{\Scal}{\mathcal{S}}
\newcommand{\supp}{{\rm supp}}
\newcommand{\id}{\mathrm{id}}
\newcommand{\kt}{{\kappa^{(t,s)}}}

\newcommand{\eps}{\varepsilon}
\newcommand{\dt}{\frac{d}{dt}}
\renewcommand{\tilde}{\widetilde}
\renewcommand{\bar}{\overline}
\renewcommand{\phi}{\varphi}
\renewcommand{\div}{\mathrm{div}_z}

\newcommand{\Zcal}{\mathcal{Z}(q,p)}
\newcommand{\Zcalinv}{\mathcal{Z}^{-1}(q,p)}

\newcommand{\Zcalinvt}{\mathcal{Z}^{-1}(t,q,p)}
\newcommand{\Zcalts}{\mathcal{Z}(t,s,q,p)}
\newcommand{\Zcalinvts}{\mathcal{Z}^{-1}(t,s,q,p)}

\newtheorem{proposition}{Proposition}
\newtheorem{theorem}{Theorem}

\newtheorem{lemma}{Lemma}
\newtheorem{corollary}{Corollary}
\newtheorem{definition}{Definition}
\theoremstyle{remark}
\newtheorem{remark}{Remark}

\title{\bf A Mathematical Justification for the Herman-Kluk Propagator}

\author{Torben SWART and Vidian ROUSSE\\
{\it \small Freie Universit\"at Berlin }}

\begin{document}
\maketitle
\begin{abstract}
A class of Fourier Integral Operators which converge to the unitary group of the Schr\"odinger equation in semiclassical limit $\eps\to 0$ is constructed. The convergence is in the uniform operator norm and allows for an error bound of order $O(\eps^{1-\rho})$ for Ehrenfest timescales, where $\rho$ can be made arbitrary small. For the shorter times of order $O(1)$, the error can be improved to arbitrary order in $\eps$. In the chemical literature the approximation is known as the Herman-Kluk propagator.
\end{abstract}
\section{Introduction}
We study approximate solutions of the semiclassical time-dependent Schr\"odinger equation
\begin{equation}\label{eq:TDSE}
i\eps\dt\psi^\eps(t)=-\frac{\eps^2}{2}{\rm\Delta}\psi^\eps(t)+V(x)\psi^\eps(t),\qquad\psi^\eps(0)=\psi^\eps_0\in L^2(\Rbb^d,\Cbb)
\end{equation}
in the semiclassical limit $\eps\to 0$. The operator $H^\eps:=-\frac{\eps^2}{2}{\rm\Delta}+V(x)$ on the right-hand side of \eqref{eq:TDSE} is the so-called Hamiltonian, a self-adjoint operator on $L^2(\Rbb^d,\Cbb)$.
It is well-known that the solution of~\eqref{eq:TDSE} can be written as
$$
\psi^\eps(t)=e^{-\frac{i}{\eps}H^\eps t}\psi^\eps_0,
$$
where the group of unitary operators $e^{-\frac{i}{\eps}H^\eps t}$ is defined by the spectral theorem.

The semiclassical parameter $\eps$ may be thought of as the quantum of action $\hbar$, but there are also situations, where $\eps$ has a different meaning. One example is provided by Born-Oppenheimer molecular dynamics, where equation~\eqref{eq:TDSE} describes the semiclassical motion of the nuclei of a molecule in the case of well-separated electronic energy surfaces and $\eps$ is the square root of the ratio of the electronic mass and the average nuclear mass. In this case, the $\eps$ in front of the time-derivative in~\eqref{eq:TDSE} is due to a rescaling of time $\tilde{t}=t/\eps$. This particular choice, the so-called ``distinguished limit'' (see~\cite{[Cole]}) produces the most interesting results in the semiclassical limit $\eps\to 0$.

To formulate our main result, we introduce the following class of Fourier Integral Operators (FIOs):
\begin{equation}\label{eq:FIO}
\Ical^\eps(\kappa^t; u)\phi(x):=\frac{1}{(2\pi\eps)^{3d/2}}\int_{\Rbb^{3d}} e^{\frac{i}{\eps}\Phi^{\kappa^t}(x,y,q,p)}
u(x,y,q,p) \phi(y)\:dq\:dp\:dy,
\end{equation}
where
\begin{itemize}
\item $\kappa^t(q,p)=\left(X^{\kappa^t}(q,p),\Xi^{\kappa^t}(q,p)\right)$ is a $C^1$-family of \emph{canonical transformations} of the classical phase space $T^\ast\Rbb^d=\Rbb^d\times\Rbb^d$,
\item $S^{\kappa^t}(q,p)$ is the associated \emph{classical action}
$$
S^{\kappa^t}(q,p)=\int\limits_0^t\left[ \dt X^{\kappa^\tau}(q,p)\cdot\Xi^{\kappa^\tau}(q,p)-(h\circ\kappa^\tau)(q,p)\right]\;d\tau,
$$
\item the complex-valued {\em phase function} is given by
\begin{align}
\Phi^{\kappa^t}(x,y,q,p)&=S^{\kappa^t}(q,p)+\Xi^{\kappa^t}(q,p)\cdot\left(x-X^{\kappa^t}(q,p)\right)-p\cdot(y-q)\nonumber\\
&\quad+\frac{i}{2}\left|x-X^{\kappa^t}(q,p)\right|^2+\frac{i}{2}\left|y-q\right|^2\label{eq:phase}
\end{align}
\item and the {\em symbol} $u$ is a smooth complex-valued function which is bounded with all its derivatives.
\end{itemize}
For this class of operators, the authors previously established an $L^2$-boundedness result, see \cite{[RousseSwart]}. The central result of this paper reads

\textbf{Theorem. }{\it 
Let $e^{-\frac{i}{\eps} H^\eps t}$ be the propagator defined by the time-dependent Schr\"o\-dinger equation~\eqref{eq:TDSE} on the time-interval $[-T,T]$ with subquadratic potential $V\in C^\infty(\Rbb^d,\Rbb)$, i.e.
$\sup_{x\in\Rbb^{d}}|\partial^\alpha_{x}V(x)|<\infty$ for all $\alpha\in\Nbb^{d} \textrm{ with }|\alpha|\geq 2$.

Then
$$
\sup_{t\in [-T,T]}\left\|e^{-\frac{i}{\eps} H^\eps t}-\Ical^\eps\left(\kappa^t; u\right)\right\|_{L^2\to L^2}\leq C(T)\eps,
$$
where $\kappa^t=(X^{\kappa^t},\Xi^{\kappa^t})$ and $u$ are uniquely given as
\begin{itemize}
\item the flow associated to the classical Hamiltonian $h(x,\xi)=\frac12|\xi|^2+V(x)$
\begin{align*}
\dt X^{\kappa^t}(q,p)&=\Xi^{\kappa^t}(q,p)&X^{\kappa^0}(q,p)=q\\
\dt\Xi^{\kappa^t}(q,p)&=-\nabla V\left(X^{\kappa^t}(q,p)\right)&\Xi^{\kappa^0}(q,p)=p
\end{align*}
and
\item the solution of the Cauchy-problem
\begin{align*}
\dt u(t,q,p)&=\frac{1}{2}u(t,q,p) \tr\left(\mathcal{Z}^{-1}\left(F^{\kappa^t}(q,p)\right)\dt \mathcal{Z}\left(F^{\kappa^t}(q,p)\right)\right)\\
u(0,q,p)&=2^{d/2}.
\end{align*}
\end{itemize}
The $\Cbb^{d\times d}$-valued function
\begin{align*}
\mathcal{Z}\left(F^{\kappa^t}(q,p)\right)&=(i\,\id\quad \id)F^{\kappa^t}(q,p)^\dagger\begin{pmatrix}-i\,\id\\ \id\end{pmatrix}\\
&=X^{\kappa^t}_q(q,p)-i X^{\kappa^t}_p(q,p)+i\Xi^{\kappa^t}_q(q,p)+\Xi^{\kappa^t}_p(q,p),
\end{align*}
depends on elements of the transposed Jacobian
\begin{align*}
F^{\kappa^t}(q,p)^\dagger=
\begin{pmatrix}
X^{\kappa^t}_q(q,p)&\Xi^{\kappa^t}_q(q,p)\\
X^{\kappa^t}_p(q,p)&\Xi^{\kappa^t}_p(q,p)\\
\end{pmatrix}
\end{align*}
of $\kappa^t$ with respect to $(q,p)$.
}

The equation for $u$ is easily solved. Its solution is the so-called Herman-Kluk prefactor
$$
u(t,q,p)=\left(\det\mathcal{Z}\left(F^{\kappa^t}(q,p)\right)\right)^\frac12,
$$ 
where the branch of the square root is chosen by continuity in time starting from $t=0$. We presented a simplified version of our main result. Theorem~\ref{theo:main_result} will essentially add three central aspects. First, we will state it for more general Hamilton operators, namely certain Weyl-quantised pseudodifferential operators. Second, for the Ehrenfest-timescale $T(\eps)=C_T\log(\eps^{-1})$ the result still holds with a slightly weaker bound. Third, the error estimate can be improved to $\eps^N$, where $N$ is arbitrary large by adding a correction of the form
$\sum_{n=1}^{N-1}\eps^n u_n$ to $u$. As $u$, the $u_n$ are solutions of explicitly solvable Cauchy-problems.

Whereas there is an abundant number of works on Fourier Integral Operators in the mathematical literature, only few of them discuss the relation between FIOs and the time-dependent Schr\"odinger-equation. The first works which apply FIOs with real-valued phase function to this problem are~\cite{[KitadaKumanoGo]} and~\cite{[Kitada]}. In this case one has to deal with the boundary value problem
$$
\textrm{Given }x,y\in\Rbb^d\textrm{, find }p\textrm{ such that }X^{\kappa^t}(y,p)=x.
$$
To get uniqueness for its solution one has either to restrict to short times $t$ or to impose very strong restrictions on the potential. The same problems are met in~\cite{[Fujiwara1]}, where Fujiwara applies a related class of operators without integral in the oscillatory kernel to the Schr\"odinger equation to justify the time-slicing approach for Feynman's path integrals.

The avoidance of this problem is the major advantage of complex-valued phase functions. In the non-semiclassical setting, Tataru shows in \cite{[Tataru]} that the unitary group of time evolution is an FIO with complex-valued phase function (different from~\eqref{eq:phase}). He also establishes that the simpler choice $u(t,q,p)=2^{d/2}$ leads to a parametrix for the non-semiclassical Schr\"odinger equation.

A class of operators related to~\eqref{eq:FIO} is used in the works~\cite{[LaptevSigal]} and~\cite{[Butler]} for the construction of approximate solutions of the semiclassical time-dependent Schr\"o\-dinger equation. In their case, the kernel consists of an integral over the momentum space in contrast to the phase-space integral in our expression
$$
\left(\tilde{\Ical}^\eps(\kappa^t;\tilde{u})\psi\right)(x)=\frac{1}{(2\pi\eps)^d}\int_{T^\ast\Rbb^d} e^{\frac{i}{\eps}\Phi^{\kappa^t}(x,y,y,p)}\tilde{u}(t,y,p)\psi(y)\;dp\,dy.
$$
Moreover, these works only allow compactly supported symbols, which enforces the truncation of the Hamiltonian in momentum. Finally there is the work of Bily and Robert~\cite{[BilyRobert]}, which treats the so-called Thawed Gaussian Approximation discussed below.

In contrast to the mathematical literature connecting time-dependent Schr\"o\-dinger equation and Fourier Integral Operators, there is an abundant number of papers in chemical journals on this topic. Nevertheless, the focus is mainly put on three approximations: the ``Thawed Gaussian Approximation'' (TGA), the ``Frozen Gaussian Approximation'' (FGA) and the Herman-Kluk expression. Confusingly, in the chemical literature both TGA and FGA do not only refer to specific algorithms but they are also used to describe whole classes of approximations. For example, the Herman-Kluk approximation is sometimes considered as an FGA, whereas the TGA refers both to the time-evolution of a coherent state and a Fourier Integral Operator. We give a short formal discussion of the most important methods in the rest of this introduction hinting to related rigorous results.

The starting point is the following identity, which holds for $\psi\in L^2(\Rbb^d,\Cbb)$:
\begin{equation}\label{eq:Overcomplete}
\psi(x)=\frac{1}{(2\pi\eps)^{d}}\int_{T^\ast\Rbb^d}g^{\eps}_{(q,p)}(x) \langle g^{\eps}_{(q,p)},\psi \rangle\;dq\,dp,
\end{equation}

where
\begin{equation}\label{eq:CS}
g^{\eps}_{(q,p)}(x)=\frac{1}{(\pi\eps)^{d/4}}e^{-|x-q|^2/2\eps}e^{ip\cdot(x-q)/\eps}
\end{equation}
denotes the coherent state centered at $(q,p)$ in phase space $T^\ast\Rbb^d$.
Within the chemical community, equation~\eqref{eq:Overcomplete} is heuristically explained as an ``expansion in an overcomplete set of Gaussians'', but the equality can be made rigorous with help of the FBI-transform, consider~\cite{[Martinez]}.
 Applying the unitary group of~\eqref{eq:TDSE} to expression~\eqref{eq:Overcomplete}, one gets the formal equality
\begin{equation}\label{eq:TDSE_CS}
\left(e^{-\frac{i}{\eps}op^\eps(h)t}\psi_0^\eps\right)(x)=\frac{1}{(2\pi\eps)^{d}}\int_{T^\ast\Rbb^d}\left(e^{-\frac{i}{\eps}op^\eps(h)t}g^{\eps}_{(q,p)}\right)(x) \langle g^{\eps}_{(q,p)},\psi_0^\eps \rangle\;dq\,dp.
\end{equation}
Hence, one expects an approximation to the solution of~\eqref{eq:TDSE} if the following approximate expression for the time-evolution of coherent states is used in~\eqref{eq:TDSE_CS}
\begin{align}
&\left(e^{-\frac{i}{\eps}op^\eps(h)t}g^\eps_{(q,p)}\right)(x)\approx\frac{1}{(\pi\eps)^{d/4}}\left[\det\left(X_q^{\kappa^t}(q,p)+iX_p^{\kappa^t}(q,p)\right)\right]^{-\frac12}\label{eq:approx_CS}\\
&\quad\times e^{\frac{i}\eps S^{\kappa^t}(q,p)}\; e^{-(x-X^{\kappa^t}(q,p))\cdot\Theta^{\kappa^t}(q,p)(x-X^{\kappa^t}(q,p))/2\eps}\; e^{i\Xi^{\kappa^t}(q,p)\cdot(x-X^{\kappa^t}(q,p))/\eps} \nonumber
\end{align}
with
$$
\Theta^{\kappa^t}(q,p)=-i\left(\Xi_q^{\kappa^t}(q,p)+i\Xi_p^{\kappa^t}(q,p)\right)\left(X_q^{\kappa^t}(q,p)+iX_p^{\kappa^t}(q,p)\right)^{-1}.
$$
In the chemical literature~\eqref{eq:approx_CS} was first derived in~\cite{[TGA]}. For rigorous mathematical results consider \cite{[Hagedorn]}, \cite{[Hagedorn1]} or~\cite{[CombescureRobert]}. As the coherent state changes its width, expression~\eqref{eq:approx_CS} and the resulting operator were baptised ``Thawed Gaussian Approximation''.

However, it turns out numerically (see e.g. the computations in~\cite{[HarabatiRostGrossmann]}) that more accurate approximations are obtained if one drops the time-dependent spreading and uses expressions like~\eqref{eq:FIO}. In the simplest case, the symbol $u\equiv 1$ is held constant in $t$, $q$ and $p$. This approximation is known as the ``Frozen Gaussian Approximation'' and holds only for times of order $O(\eps)$, see the remark after Theorem~\ref{theo:main_result}. To get to the longer times of order $O(1)$, the more sophisticated choice of $u(t,q,p)$ as the Herman-Kluk prefactor is needed, see~\cite{[HermanKluk]} for the original work and~\cite{[Kay1]} and~\cite{[Kay2]} for works, which are methodically related to our presentation. Moreover, the latter of them presents the first derivation of the higher order corrections.

\subsection*{Organisation of the paper and notation}

The paper is organised in the following way. Section 2 will set the stage for the discussion of our approximation. Here we will recall central definitions and results on Fourier Integral Operators, first and foremost their definition and well-definedness on the functions of Schwartz class as well as their bound as operators acting on $L^2(\Rbb^d,\Cbb)$, see Definition~\ref{defi:FIO} and Theorem~\ref{theo:FIO}. Most of the results of this section can be found in~\cite{[RousseSwart]} and we refer the reader to that paper for a more detailed discussion and motivation of them. In Section 3 we will prove results on the composition of Weyl-quantised pseudo-differential operators and Fourier Integral Operators, see~Proposition~\ref{prop:comp_PDO}. Moreover, we will investigate the time-derivative of a $C^1$-family of Fourier Integral Operators in Proposition~\ref{prop:comp_time}. These results will lead to our main result, which we will state in Theorem~\ref{theo:main_result}.

We close this introduction by a short discussion of the notation. Throughout this paper, we will use standard multiindex notation. Vectors will always be considered as column vectors. The inner product of two vectors $a, b\in\Rbb^d$ will be denoted as $a\cdot b=\sum_{j=1}^da_j b_j$ and extended to vectors $a,b\in\Cbb^d$ by the same formula. The transpose of a matrix $A$ will be $A^\dagger$, whereas $A^\ast:=\bar{A}^\dagger$ denotes the adjoint and finally $e_j$ will stand for the $j$th canonical basis vector of $\Rbb^d$ or $\Cbb^d$.

For a differentiable mapping $F\in C^1(\Rbb^d,\Cbb^d)$, we will use both $(\partial_x F)(x)$ and $F_x(x)$ for the transpose of its Jacobian at $x$, i.e. $((\partial_x F)(x))_{jk}=(F_x(x))_{jk}=(\partial_{x_j} F_k)(x)$. This leads to the identity $\partial_x (F\cdot G)=G_x F+F_x G$ for $F,G\in C^1(\Rbb^d,\Cbb^d)$. The Hessian matrix of a mapping $F\in C^2(\Rbb^d,\Cbb)$ will be denoted by $\textrm{Hess}_{x}F(x)$.

For the sake of better readability of the formulae, we will be somewhat sloppy with respect to the distinction between functions and their values. As a crucial example, we will write $(x-X^\kappa(q,p))v$ for the function $(x,y,q,p)\mapsto(x-X^\kappa(q,p))v(x,y,q,p)$.

When dealing with canonical transformations, we introduce the following notations for a complex linear combinations of the components:
\begin{align*}
Z^\kappa(q,p)&:=\left(\Theta^x\right)^{\frac12}X^\kappa(q,p)+i\left(\Theta^x\right)^{-\frac12}\Xi^\kappa(q,p)\\
\bar{Z}^\kappa(q,p)&:=\left(\Theta^x\right)^{\frac12}X^\kappa(q,p)-i\left(\Theta^x\right)^{-\frac12}\Xi^\kappa(q,p).
\end{align*}
We want to point out that $\bar{Z}^\kappa(q,p)$ is not the complex conjugate of $Z^\kappa(q,p)$ for non-real matrices $\Theta^x$. The matrix square root of a positive definite matrix will always be chosen as the unique positive definite square root. We want to point out that both the determinant of this matrix-square root and the square root of a determinant will appear in this paper.

We define $z:=\Theta^y q+ip$, $\partial_z:=\left(\Theta^y\right)^{-1}\partial_q-i\partial_p$ and
$$
\div X(q,p)=\sum_{k=1}^{d}\left(\Theta^y\right)^{-1}_{jk}\partial_{q_k} X_j(q,p)
-i\sum\limits_{j=1}^{d}\partial_{p_j}X_j(q,p)
$$
for functions $X\in C^1(\Rbb^{2d},\Cbb^d)$, regardless whether they are row or column vectors. With these definitions the identity $\div X(q,p)=\tr  X_z(q,p)$ still holds. Finally, we mention that the expression $\dt X^\kt(q,p)\cdot\Xi^\kt(q,p)$ denotes the inner product of $\dt X^\kt(q,p)$ and $\Xi^\kt(q,p)$.

\subsection*{Acknowledgement}
The authors want to thank Caroline Lasser for many profitable discussions and valuable comments.

\section{Canonical Transformations and Fourier Integral Operators}\label{section:recall}
In this section, we specialise the central definitions and results of~\cite{[RousseSwart]} to the case of Hamiltonian flows.
\subsection{Symbol classes and canonical transformations}
The definition of our FIOs involves two fundamental objects. One of them is a smooth complex-valued function, the so-called symbol. The following definition deviates from~\cite{[RousseSwart]} by the additional $\eps$-dependence.

\begin{definition}[Symbol class]
Let  ${\bf m}=(m_j)_{1\leq j\leq J}\in\Rbb^J$ and ${\bf d}=(d_j)_{1\leq j\leq J}\in\Nbb^J$. We say that $u:\:]0,1]\times\Rbb^{\bf |d|}\to\Cbb^N$ is a \textbf{symbol of class $S[{\bf m}; {\bf d}]$}, if there is $\eps_0<1$, such that $u^\eps\in C^\infty(\Rbb^{\bf |d|},\Cbb^N)$ for all $\eps\leq\eps_0$ and the following quantities are finite for any $k\geq0$
\begin{equation}
M^m_k[u]:=\sup_{\eps\leq\eps_0}\max_{|\alpha|=k}\sup_{z\in\Rbb^{|{\bf d}|}}\left|\prod_{j=1}^J\langle z_j\rangle^{-m_j}\partial_z^\alpha u^\eps(z)\right|,
\end{equation}
where $\langle z\rangle:=\sqrt{1+|z|^2}$.
We extend this definition by setting
\begin{equation}
S[+\infty;{\bf d}]=\bigcup_{m_1\in\Rbb}\ldots\bigcup_{m_J\in\Rbb}S[(m_1,\ldots,m_J);{\bf d}].
\end{equation}
\end{definition}

The second central object in the definition of a Fourier Integral Operator is a canonical transformation of the classical phase space.
\begin{definition}[Canonical transformation]
Let $\kappa(q,p)=(X^\kappa(q,p),\Xi^\kappa(q,p))$ be a diffeomorphism of $T^\ast\Rbb^d=\Rbb^d\times\Rbb^d$. We represent its differential by the following Jacobian matrix
\begin{equation} \label{eq:Jacobian}
F^\kappa(q,p)=
\begin{pmatrix}
X^\kappa_q(q,p)^\dagger & X^\kappa_p(q,p)^\dagger \\
\Xi^\kappa_q(q,p)^\dagger & \Xi^\kappa_p(q,p)^\dagger
\end{pmatrix}.
\end{equation}
$\kappa$ is said to be a \textbf{canonical transformation} if $F^\kappa(q,p)$ is symplectic for any $(q,p)$ in $T^\ast\Rbb^d$, i.e.
\begin{equation*}
F^\kappa(q,p)\in {\rm Sp}(2d):=\left\{S\in {\rm Gl}(2d)\middle| S^\dagger J S=J \right\}
\quad\textrm{ with }\quad J:=
\begin{pmatrix}
0&\id\\
-\id&0
\end{pmatrix}.
\end{equation*}
\end{definition}

To get good properties for our operators, we need to restrict the class of canonical transformations under consideration.
\begin{definition}[Canonical transformation of class $\Bcal$]
A canonical transformation $\kappa$ of $T^\ast\Rbb^d$ is said to be \textbf{of class $\Bcal$} if $F^\kappa\in S[0;2d]$. A time-dependent family of canonical transformations $\kappa^t$ will be called \textbf{of class $\Bcal$} in $[-T,T]$ if it is pointwise continuously differentiable with respect to time and we have for all $k\geq 0$
$$
\sup_{t\in [-T,T]}M_k^0\left[F^{\kappa^t}\right]<\infty\quad\textrm{ and }\quad\sup_{t\in [-T,T]}M_k^0\left[\dt F^{\kappa^t}\right]<\infty.
$$
In particular $F^{\kappa^t}$ and $\dt F^{\kappa^t}$ are of class $S[0;2d]$ pointwise for $t\in [-T, T]$.
\end{definition}

We also have to restrict the Hamiltonians we use.
\begin{definition}
A time-dependent Hamiltonian $h\in C(\Rbb,C^\infty(\Rbb^{2d},\Cbb))$ is called \textbf{subquadratic}, if
\begin{equation}
\sup_{-T\leq t\leq T}\sup_{(x,\xi)\in\Rbb^d\times\Rbb^d}\Vert\partial^\alpha_{(x,\xi)}h(t,x,\xi)\Vert_{L^\infty}
\end{equation}
is finite for all $|\alpha|\geq 2$ and $T>0$. It is called \textbf{sublinear}, if the quantity is finite for all $|\alpha|\geq 1$.
\end{definition}

The next result will investigate the relation between classical Hamiltonians and the flows they generate.
\begin{proposition}\label{prop:subquadratic}
If $h\in C(\Rbb,C^\infty(\Rbb^{2d},\Cbb))$ is a time-dependent subquadratic Ha\-miltonian, the Hamiltonian flow $\kt$ generated by $h$,
\begin{align}\label{eq:HSt}
\dt \kt=J \nabla_{(x,\xi)} h(t,\kt),\quad \kappa^{(s,s)}=\id
\end{align}
is a family of canonical transformations of class $\Bcal$ in $[-T,T]$. Moreover, every Hamiltonian flow of class $\Bcal$ is generated by a subquadratic Hamiltonian.

Under the additional assumption $\|\partial^\alpha_{(x,\xi)} h\|_{L^\infty(\Rbb\times\Rbb^{2d})}<\infty$ for all $2\leq|\alpha|\leq n_0+2$, we have 
$$
\sup_{|t-s|<T(\eps)}M_k^0\left[F^\kt\right]\leq C_k (2C_T)^k\left|\log\eps\right|^{k}\eps^{-2 K_0^hC_T},
$$
for all $k\leq n_0$ on the Ehrenfest timescale $T(\eps)=C_T\log\eps^{-1}$, where
$$
K^h_k(T)=\max_{|\alpha|=k}\sup_{-T\leq t\leq T}\sup_{(x,\xi)\in\Rbb^d\times\Rbb^d}\Vert\partial^\alpha_{(x,\xi)}\mathrm{Hess}_{(x,\xi)}h(t,x,\xi)\Vert.
$$
\end{proposition}

\begin{proof}
The basic identity follows by differentiating~\eqref{eq:HSt} with respect to $(q,p)$:
\begin{equation}\label{eq:dtF} 
\dt F^{\kappa^{(t,s)}}(q,p)=J\textrm{Hess}_{(x,\xi)}h(t,\kappa^{(t,s)}(q,p))F^{\kappa^{(t,s)}}(q,p).
\end{equation}
The fundamental theorem of calculus gives
\begin{align*}
\left\|F^{\kappa^{(t,s)}}(q,p)\right\|\leq&\left|\int\limits_{s}^{t}\left\|\textrm{Hess}_{(x,\xi)}h(\tau,\kappa^{(\tau,s)}(q,p))\right\|\: \left\|F^{\kappa^{(\tau,s)}}(q,p)\right\|\;d\tau\right|+\|\id\|\\
\leq& K_0^h(T)\left|\int\limits_{s}^{t} \left\|F^{\kappa^{(\tau,s)}}(q,p)\right\|\;d\tau\right|+1.
\end{align*}

Hence, the bound for the Jacobian follows from Gronwall's Lemma:
$$
\left\|F^{\kappa^{(t,s)}}(q,p)\right\|\leq e^{K_0^h(T) |t-s|}.
$$
For the derivatives we have
\begin{align*}
&\left\|\partial^{\alpha}_{(q,p)}F^{\kappa^{(t,s)}}(q,p)\right\|
\leq
\left|\int\limits_{s}^{t}K_0^h(T) \left\|\partial^{\alpha}_{(q,p)}F^{\kappa^{(\tau,s)}}(q,p)\right\|\;d\tau\right|\\
&\qquad+\left|\int\limits_{s}^{t}
\sum_{\beta<\alpha}\binom{\alpha}{\beta}
\left\|\partial_{(q,p)}^{\alpha-\beta}\left[\textrm{Hess}_{(x,\xi)}h(\tau,\kappa^{(\tau,s)}(q,p))\right]
\partial_{(q,p)}^{\beta}F^{\kappa^{(t,s)}}(q,p)\right\|
\;d\tau\right|,
\end{align*}
so Gronwall's Lemma provides inductively
$$
\left\|\partial^{\alpha}_{(q,p)}F^{\kappa^{(t,s)}}(q,p)\right\|\leq C_k (2T)^ke^{K_0^h(T) |t-s|},
$$
where $C_k$ depends on $K_l^h(T)$ for $l\leq k$.
The result for the Ehrenfest timescale follows by substituting $T(\eps)=C_T\log(\eps^{-1})$ into this expression.

Now consider a Hamiltonian flow of class $\Bcal$. The identity~\eqref{eq:dtF} gives
\begin{align*}
J \left(\dt F^{\kappa^{(t,s)}}(q,p)\right)J\left(F^{\kappa^{(t,s)}}(q,p)\right)^\dagger J=\textrm{Hess}_{(q,p)}h(t,\kappa^{(t,s)}(q,p))
\end{align*}
Hence, $h$ is subquadratic, as $\dt F^{\kappa^{(t,s)}}$ is of class $S[0;2d]$ by definition.
\end{proof}

\begin{remark}\hspace*{\fill}\label{rem:ehrenfest}
\begin{enumerate}
\item It is easy to show that there is an equivalence between linear Hamiltonian flows and quadratic Hamiltonians. Every quadratic Hamiltonian generates a linear flow and every Hamiltonian linear flow
$\kt(q,p)=M(t,s)\binom{q}{p}$, $M(\cdot,s)\in C^1(\Rbb,\mathrm{Sp}(2d))$
is generated by a quadratic Hamiltonian, namely
$$
h_M(t,q,p)=\frac{1}{2}(q\: \:p)\left(J\left(\dt M(t,s)\right)JM^\dagger(t,s)J\right)
\binom{q}{p}.
$$

\item By estimating the logarithm, we can have a bound of the form
\begin{equation}\label{eq:bound_ehrenfest}
\sup_{|t-s|<T(\eps)}M_k^0\left[F^\kt\right]\leq C'_k(C_T)\eps^{-\rho(C_T)},
\end{equation}
for the Ehrenfest timescale, where $\rho(C_T)<\rho_0$ for any $\rho_0>0$, if $C_T$ is chosen small enough.

\item From now on, all considered canonical transformations are assumed to be of class $\Bcal$. 
\end{enumerate}
\end{remark}

An important quantity associated with a canonical transformation is the so-called action.
\begin{definition}[Action]
Let $\kappa(q,p)=(X^\kappa(q,p),\Xi^\kappa(q,p))$ be a canonical transformation of $T^\ast\Rbb^d$. A real-valued function $S^\kappa$ is called an \textbf{action associated to $\kappa$} if it fulfills
\begin{equation} \label{eq:action}
S^\kappa_q(q,p)=-p+X^\kappa_q(q,p)\Xi^\kappa(q,p) ,\quad S^\kappa_p(q,p)=X^\kappa_p(q,p)\Xi^\kappa(q,p).
\end{equation}
\end{definition}
\begin{remark}\hfill
\begin{enumerate}
\item An action associated to a canonical transformation is only defined up to an additive constant.
If we consider a time-dependent family of canonical transformations $\kappa^t$, we will choose this time-dependent constant such that $S^{\kappa^t}(q,p)$ is $C^1$ with respect to time.

\item If $\kappa^{(t,s)}$ is induced by a Hamiltonian $h(t,x,\xi)$, the action of classical mechanics
$$
S_{cl}^{\kappa^{(t,s)}}(q,p)=\int\limits_s^t\left( \dt X^{\kappa^\tau}(q,p)\cdot \Xi^{\kappa^\tau}(q,p) -h\left(\tau,\kappa^{(\tau,s)}(q,p)\right)\right)\;d\tau
$$
is an action in the sense of this definition. In this case, we use the convention $S^{\kappa^{(s,s)}}(q,p)=0$, where $S^\kt(q,p)$ is now considered as a function of $t$. We cannot assume $S^\id(q,p)=0$, as the case $h(t,x,\xi)=h(t)$ shows.
\end{enumerate}
\end{remark}
\subsection{Definition of FIOs and continuity results}
In this section, we will define the operators we will use to approximate of the propagator of the Schr\"odinger equation.

\begin{definition}[Fourier Integral Operator]\label{defi:FIO}
For $u\in S[(+\infty,m^p);(3d,d)]$, a Schwartz-class function $\varphi\in\Scal(\Rbb^d,\Cbb)$ and $n>m^p+d$, we define
\begin{align}
&\left[\Ical^\eps(\kappa;u;\Theta^x,\Theta^y)\varphi\right](x):=\label{eq:defi_FIO}\\
&\qquad\frac{1}{(2\pi\eps)^{3d/2}}\int_{\Rbb^{3d}}e^{\frac{i}{\eps}\Phi^\kappa(x,y,q,p;\Theta^x,\Theta^y)}(L_y^\dagger)^n[u(x,y,q,p)\varphi(y)]\;dq\, dp\, dy\nonumber,
\end{align}
where
\begin{itemize}
\item $\Theta^x$ and $\Theta^y$ are complex symmetric matrices (i.e. $\Theta=\Theta^\dagger$) with positive definite real part,
\item the complex-valued phase-function is given by
\begin{align*}
\Phi^\kappa(x,y,q,p;\Theta^x,\Theta^y)&=S^\kappa(q,p)-p\cdot (y-q) + \Xi^\kappa(q,p)\cdot(x-X^\kappa(q,p))\\
&\quad+\frac{i}{2}(y-q)\cdot \Theta^y (y-q)\\
&\quad+\frac{i}{2} (x-X^\kappa(q,p))\cdot \Theta^x (x-X^\kappa(q,p))
\end{align*}
\item and the differential operator $L_y$ is defined by
$$
L_y=\frac{1}{1+|\nabla_y\Phi^\kappa(x,y,q,p;\Theta^x,\Theta^y)|^2}\left[1-i\eps\nabla_y\overline{\Phi^\kappa(x,y,q,p;\Theta^x,\Theta^y)}\cdot\nabla_y\right].$$
\end{itemize}
\end{definition}
The operator $L^\dagger_y$, the Banach-space adjoint of $L_y$, i.e.
$$
\int_{\Rbb^d}v(y)[L_y^\dagger u](y)\;dy=\int_{\Rbb^d}[L_yv](y)u(y)\; dy,
$$
makes the integrals absolutely convergent. Moreover, it is constructed such that expression~\eqref{eq:defi_FIO} coincides for different $n$ with $n>m^p+d$ (especially with $n=0$, if $m^p<-d$).

The following theorem combines the central results of~\cite{[RousseSwart]}.
\begin{theorem} \label{theo:FIO}\hfill
\begin{enumerate}
\item If $u\in S[+\infty;4d]$, $\Ical^\eps(\kappa;u;\Theta^x,\Theta^y)$ sends $\Scal(\Rbb^d,\Cbb)$ into itself and is continuous.
\item If $u\in S[0;4d]$, $\Ical^\eps(\kappa;u;\Theta^x,\Theta^y)$ can be extended in a unique way to a linear bounded operator $L^2(\Rbb^d,\Cbb)\to L^2(\Rbb^d,\Cbb)$ and there exists a constant $C(M_0^\kappa,\Theta^x,\Theta^y)$ such that
\begin{equation} \label{eq:full}
\left\| \Ical^\eps(\kappa;u;\Theta^x,\Theta^y)\right\|_{L^2\to L^2}\leq C(M_0^\kappa;\Theta^x,\Theta^y) \sum_{|\alpha|\leq 4d+1}\|\partial^\alpha_{(x,y)}u\|_{L^\infty}.
\end{equation}
In the special case where $u\in S[0;2d]$ is independent of $(x,y)$, we have
\begin{equation} \label{eq:corfull}
\left\| \Ical^\eps(\kappa;u;\Theta^x,\Theta^y)\right\|_{L^2\to L^2}\leq 2^{-d/2}\det\left(\Re\Theta^x\Re\Theta^y\right)^{-\frac14}\Vert u\Vert_{L^\infty}.
\end{equation}
\end{enumerate}
\end{theorem}
\begin{remark}\label{remark:identity}\hfill
\begin{enumerate}
\item The dependence of $C(M_0^\kappa;\Theta^x,\Theta^y)$ on $M_0^\kappa,\Theta^x$ and $\Theta^y$ can be made more explicit. Consider~{\rm \cite{[RousseSwart]}} for the precise expression.
\item There is an analogous result for Weyl-quantised pseudodifferential operators
$$
(op^\eps(h)\psi)(x):=\frac{1}{(2\pi\eps)^d}\int_{T^\ast\Rbb}e^{\frac{i}{\eps}\xi\cdot(x-y)}
h\left(\frac{x+y}{2},\xi\right)\psi(y)\:dy\:d\xi,
$$
see for example~{\rm\cite{[Martinez]}}:
\begin{enumerate}
\item If $h\in S[+\infty;2d]$,  $op^{\eps}(h)$ sends $\Scal(\Rbb^d,\Cbb)$ into itself and is continuous.
\item If $h\in S[0;2d]$, $op^{\eps}(h)$ can be extended in a unique way to a linear bounded operator $L^2(\Rbb^d,\Cbb)\to L^2(\Rbb^d,\Cbb)$ with $\eps$-independent norm
$$
\left\|op^\eps(h)\right\|_{L^2\to L^2}\leq C\sum_{|\alpha|\leq 2d+1}\|\partial^\alpha_{(x,\xi)}h\|_{L^\infty}.
$$
\end{enumerate}
The second part is the famous Calder\'on-Vaillancourt Theorem.
\item We have $\Ical^\eps\left(\id;\det(\Theta^x+\Theta^y)^{\frac12};\Theta^x,\Theta^y\right)=\id$,
compare the appendix for the correct choice of the square root.
\end{enumerate}
\end{remark}

\section{Composition with PDOs and time-derivatives}
The standard approach in the field of asymptotic analysis consists in a two step procedure. First, one constructs an {\em asymptotic solution} $U^\eps_{N}(t,s)\psi^\eps_s$ of order $O(\eps^{N+1})$, i.e. a function which fulfills
\begin{equation}
\left(i\eps\dt -H^\eps(t)\right)U^\eps_{N}(t,s)\psi^\eps_s=\eps^{N+1} R_N^\eps(t,s)\psi^\eps_s.\label{eq:asympt_solution}
\end{equation}
If one can establish an $\eps$-independent bound on the remainder $R_N^\eps(t,s)$, the asymptotic solution can be turned into an approximate solution of the unitary group with help of a special version of Gronwall's Lemma, (see for example Lemma 2.8 in~\cite{[Hagedorn1]} for the strategy of the proof):
\begin{lemma}\label{lemma:gronwall}
Let $U^\eps(t,s)$ be the propagator of the time-dependent Schr\"odinger-equation
$$
\left(i\eps\dt-H^\eps(t)\right)\psi^\eps(t)=0,\qquad\psi^\eps(s)=\psi^\eps_s\in D\subset L^2(\Rbb^d,\Cbb)
$$
for some family of self-adjoint operators $H^\eps(t)$ with common domain $D$.
Moreover, for some $T>0$ and $-T\leq t,s \leq T$ let $U_{N}^\eps(t,s)$ be a family of bounded operators, which is strongly differentiable with respect to $t$, leaves the domain of $H^\eps(t)$ invariant and which fulfills
$$
i\eps\dt U_{N}^\eps(t,s)\psi^\eps(s)-H^\eps(t)U_{N}^\eps(t,s)\psi^\eps(s)=\eps^{N+1}R_N^\eps(t,s)\psi^\eps(s)
$$
with $U_{N}^\eps(s,s)=\id$. If $\|R_N^\eps(t,s)\|_{L^2\to L^2}< \infty$ for all $-T\leq t,s\leq T$, we have
$$
\left\|U_{N}^\eps(t,s)-U^\eps(t,s)\right\|_{L^2\to L^2}\leq\eps^{N}\left|\int\limits_{s}^{t}\|R_N^\eps(\tau,s)\|_{L^2\to L^2}\;d\tau\right|.
$$
\end{lemma}
In this section, we state the intermediate results needed for the construction of the asymptotic solution.

In Proposition~\ref{prop:comp_PDO}, we show using Weyl-quantisation that the composition of differential operators with Fourier Integral Operators is again an FIO. Moreover, we give an asymptotic expansion of the symbol of the new FIO, whose terms but for the last are $x$-independent. This is important, as $x$-dependence of the symbol may be converted to $\eps$-dependence, which can be seen from Lemma~\ref{lemma:part_int}. Proposition~\ref{prop:comp_time} deals with the time-derivative of a family of FIOs. Finally, we will establish an uniqueness result for symbols and canonical transformations in Proposition~\ref{prop:uniqueness}.

\subsection{Statement of intermediate results}\label{section:composition}
To state our results, we need the matrix~$\Zcal=Z^\kappa_z(q,p)\left(\Theta^x\right)^{\frac12}$, which already appeared as $\mathcal{Z}\left(F^\kappa(q,p)\right)$ in the statement of our main result in the introduction. We justify this abuse of notation by better readability of the formulae presented here. The invertibility of $\Zcal$, which is implicitly claimed in the following statements, is shown in Lemma~\ref{lemma:Z_inv}.

The composition result reads:
\begin{proposition}\label{prop:comp_PDO}
Let $h\in S[m_h;2d]$ be polynomial in $\xi$ and $u\in S[m_u;2d]$. Then we have
$$
op^\eps(h)\Ical^\eps(\kappa;u;\Theta^x,\Theta^y)=
\Ical^\eps\left(\kappa;\sum\limits_{n=0}^{N}\eps^n v_n;\Theta^x,\Theta^y\right)
+\eps^{N+1}\Ical^\eps\left(\kappa; v^\eps_{N+1};\Theta^x,\Theta^y\right)
$$
as operators from $\Scal(\Rbb^d,\Cbb)$ to $\Scal(\Rbb^d,\Cbb)$.

$v_{n}\in S[m_u+m_h;2d], n\leq N$ and $v^\eps_{N+1}\in S[(m_h,m_u+m_h);(d,2d)]$ are given by
\begin{align*}
v_n(q,p)&=L_n[h;\kappa;\Theta^x,\Theta^y]u(q,p)\\
v^\eps_{N+1}(x,q,p)&=L^\eps_{N+1}[h;\kappa;\Theta^x,\Theta^y]u(q,p),
\end{align*}
where $L_n[h;\kappa;\Theta^x,\Theta^y]$ and $L^\eps_{N+1}[h;\kappa;\Theta^x,\Theta^y]$ are linear differential operators of order $n$ and $N+1$ in $(q,p)$. 
The coefficients of the $L_n[h;\kappa;\Theta^x,\Theta^y]$ are rational functions, with a numerator depending on derivatives of $h$ from order $n$ to $2n$ and derivatives of $F^\kappa(q,p)$ of order $\leq n$, and a denominator of the form $\det(\Zcal)^m$ for some $m>0$. $L^\eps_N[h;\kappa;\Theta^x,\Theta^y]$ is of the same form depending on derivatives of $h$ from order $N$ to $2N+1$ and derivatives of $F^\kappa(q,p)$ of order $\leq 2N+1$.

The explicit expressions for $v_0,v_1$ and $v_2$ are
\begin{align}
v_0(q,p)&=u(q,p)(h\circ\kappa)(q,p) \label{eq:v0_PDO}\\
\nonumber\\
v_1(q,p)\label{eq:v1_PDO}
&=-\div\left(\left((h_{x}+i\Theta^xh_{\xi})\circ\kappa(q,p)\right)^\dagger\Zcalinv u(q,p)\right)\\
&\quad+u(q,p)\frac{1}{2}\tr\left(\Zcalinv \partial_{z}((h_x+i\Theta^xh_\xi)\circ\kappa(q,p))\right)\nonumber\\
\nonumber\\
v_2(q,p)&=L_2[h_{\geq 3};\kappa;\Theta^x,\Theta^y]u(q,p)\\
+\frac{1}{2}\sum\limits_{k=1}^{d}&\;\div\left(
u(q,p)\partial_{z_k}\left[((\partial_x+i\Theta^x\partial_\xi)^2h)\circ\kappa(q,p)\Zcalinv e_k)\right]^\dagger \Zcalinv
\right)\nonumber
\end{align}
The linear partial differential operator $L_2[h_{\geq 3};\kappa;\Theta^x,\Theta^y]$ depends on derivatives of $h$ of order $3$ and $4$.
\end{proposition}
\begin{remark}
As the coefficients of the differential operators are rational functions of the form described in the statement, a bound for the elements of $F^\kappa(q,p)$ and their derivatives of the form $C\eps^{-\rho}$ gives a bound for the coefficients of the form $C'\eps^{-M\rho}$ for some $M\in\Nbb$, where $C'$ depends on derivatives of $h$.
\end{remark}

The second result of this section will investigate the time-derivative of a family of FIOs.
In the case of a time-dependent family of canonical transformations, we have the following result:
\begin{proposition}\label{prop:comp_time}
Let $u\in C(\Rbb,S[(m_q,m_p);(d,d)])$ be a family of time-dependent symbols with $u(\cdot,q,p)\in C^1(\Rbb,\Cbb)$ and $(\dt u)(t,\cdot,\cdot)\in S[(m_q,m_p);(d,d)]$, $\kappa^t$ a family of canonical transformations of class $\Bcal$, $S^{\kappa^t}$ an action associated to $\kappa^t$, $\Theta^x\in C^1(\Rbb,{\rm Gl}(d))$ a family of complex symmetric matrices with positive definite real part and $\Theta^y$ complex symmetric with positive definite real part.
We have
$$
i\eps\dt \Ical^\eps\left(\kappa^t;u;\Theta^x(t),\Theta^y\right)=
\Ical^\eps\left(\kappa^t;\sum_{n=0}^2\eps^nv_n;\Theta^x(t),\Theta^y\right)
$$
with
\begin{align}
v_{0}(t,q,p)&=u(t,q,p)\left(-\dt S^{\kappa^t}(q,p)
+\dt X^{\kappa^t}(q,p)\cdot\Xi^{\kappa^t}(q,p)\right)\label{eq:v0_time}\\
v_{1}(t,q,p)
&=i\dt u(t,q,p)\label{eq:v1_time}\\
&\quad+\div\left(\left(\dt \Xi^{\kappa^t}(q,p)-i\Theta^x(t)\dt X^{\kappa^t}(q,p)\right)^\dagger \Zcalinvt u(t,q,p)\right)\nonumber\\
&\quad-\frac{i}{2}u(t,q,p)\tr\left(\Zcalinvt X^\kappa_z(q,p)\dt\Theta^x(t)\right)\nonumber,\\
v_2(t,q,p)&=-\sum\limits_{k=1}^{d}\div\left(\partial_{z_k}\left(\dt\Theta^x(t)\Zcalinvt e_k\,u(q,p)\right)^\dagger \Zcalinvt\right),
\end{align}
where $v_{0},v_1,v_2\in C^0\left(\Rbb,S[(m_q,m_p);(d,d)]\right)$.
\end{proposition}

\begin{remark}
In both propositions, the case of a linear canonical transformation, a quadratic symbol $h$ and a constant symbol $u$ results in $v_n=0$ for $n\geq 2$. This will result in the exactness of the Herman-Kluk expression for quadratic Hamiltonians.
\end{remark}

Finally, we have the following uniqueness result.
\begin{proposition}\label{prop:uniqueness}
Let $\kappa_1$ and $\kappa_2$ be two canonical transformations of class $\Bcal$ and $u,v\in S[0;2d]$. If
$$
\lim_{\eps\to 0}\left\|\Ical^\eps(\kappa_1;u;\Theta^x,\Theta^y)-\Ical^\eps(\kappa_2;v;\Theta^x,\Theta^y)\right\|_{L^2\to L^2}=0,
$$
then $u=v$ and $\kappa_1(q,p)=\kappa_2(q,p)$ for all $(q,p)\in\supp u$.
\end{proposition}

\subsection{Some auxiliary results}
The proofs of these results will strongly rely on results about conversion of $x$-dependence to $\eps$-dependence. We start with the following auxiliary result.

\begin{lemma}\label{lemma:Z_inv}
We have
\begin{align*}
i\Phi^\kappa_{z}(x,y,q,p;\Theta^x,\Theta^y)=Z_z^\kappa(q,p)\left(\Theta^x\right)^{\frac12}(x-X^\kappa(q,p))=:\Zcal(x-X^\kappa(q,p)).
\end{align*}
$\Zcal=(i\left(\Theta^y\right)^{-1}\quad\id)(F^\kappa(q,p))^\dagger(-i\Theta^x\quad\id)^\dagger$ is invertible and its inverse $\Zcalinv$ is in the class $S[0;2d]$.
\end{lemma}
\newpage
\begin{remark}~\label{rem:Z_const}
\begin{enumerate}
\item We recall that
$$
Z^\kappa_z(q,p)=\left(\left(\Theta^y\right)^{-1}\partial_q-i\partial_p\right)\left(\left(\Theta^x\right)^{\frac12} X^\kappa(q,p)+i\left(\Theta^x\right)^{-\frac12}\Xi^\kappa(q,p)\right).
$$
\item Obviously, $\Zcal$ depends on $q$ and $p$ only via the elements of $F^\kappa(q,p)$. For better readability, we do not explicitly denote this dependence. Moreover, we drop the dependence on $\Theta^x$ and $\Theta^y$ in the notation.
\item For a linear canonical transformation $\kappa(q,p)=M\begin{pmatrix}q\\p\end{pmatrix}$,
with $M\in {\rm Sp}(2d)$, we have $F^\kappa(q,p)=M$, so $\Zcal=(i\,\id\;\;(\Theta^y)^{-1})M^\dagger(-i\Theta^x\;\;\id)^\dagger$ is constant with respect to $(q,p)$.
\end{enumerate}

\end{remark}
\begin{proof}
The derivatives of $\Phi^\kappa(x,y,q,p;\Theta^x,\Theta^y)$ with respect to $q$ and $p$ are
\begin{align*}
\Phi^\kappa_q(x,y,q,p;\Theta^x,\Theta^y)&=[\Xi^\kappa_{q}-iX_{q}^\kappa\Theta^x](q,p)(x-X^\kappa(q,p))-i\Theta^y(y-q)\\
\Phi^\kappa_p(x,y,q,p;\Theta^x,\Theta^y)&=[\Xi^\kappa_{p}-iX_{p}^\kappa\Theta^x](q,p)(x-X^\kappa(q,p))-(y-q),
\end{align*}
which gives the identity for $\Zcal$. Obviously, $\Zcal$ inherits its symbol class from
$F^\kappa(q,p)$. Moreover, we have
\begin{align*}
&\quad\Zcal\left(\Re\Theta^x\right)^{-1}\Zcal^\ast\\
&=2\Re\left(\Theta^y\right)^{-1}+\left(\Lambda\left(\Theta^x\right)F^\kappa(q,p)
\begin{pmatrix}
i\left(\Theta^y\right)^{-1}\\
-\id
\end{pmatrix}\right)^*\left(\Lambda\left(\Theta^x\right)F^\kappa(q,p)\begin{pmatrix}
i\left(\Theta^y\right)^{-1}\\
-\id
\end{pmatrix}\right)
\end{align*}
with
$$
\Lambda(\Theta)=
\begin{pmatrix}
(\Re\Theta)^{1/2}&0\\
(\Re\Theta)^{-1/2}\Im\Theta&(\Re\Theta)^{-1/2}\\
\end{pmatrix}.
$$
Hence, by the superadditivity of the determinant for positive definite hermitian matrices, $\det \Zcal$ is uniformly bounded away from $0$ for all $q$ and $p$, so by 
its expression via the formula of minors, $\Zcalinv\in S[0;2d]$, as $\Zcal$ is.
\end{proof}

We introduce the following notation:
\begin{definition}
Two symbols $u,v\in S[+\infty;4d]$ are said to be \textbf{equivalent with respect to $\kappa$} if
$$
\Ical^\eps(\kappa;u;\Theta^x,\Theta^y)=\Ical^\eps(\kappa;v;\Theta^x,\Theta^y)
$$
as operators from $\Scal(\Rbb^d,\Cbb)$ to $\Scal(\Rbb^d,\Cbb)$. In this case we write $u\sim v$.
\end{definition}

The central technical identity, on which most of the following results rely, is contained in the following lemma, which is a special case of Lemma 5 in~\cite{[RousseSwart]}. Here we present a proof based on definition~\ref{defi:FIO}, whereas in {\rm \cite{[RousseSwart]}} an alternative definition based on a smoothing of the oscillatory integrals is used.
\begin{lemma}\label{lemma:part_int}
Let $u\in S[(m^x,m^q,m^p);(d,d,d)]$.
Then
$$
(x_{j}-X^\kappa_{j}(q,p))u\sim \eps v,
$$
where $v\in S[(m^x,m^q,m^p);(d,d,d)]$ is given by
\begin{align*}
v(x,q,p)
&=-\div\left(e_{j}^\dagger \Zcalinv u(x,q,p)\right)
\end{align*}
\end{lemma}
\begin{proof}
We introduce a countable, locally finite partition of unity $\chi_{l}(q,p)$ with 
$\sup_{N,q,p}\left|\sum_{l=0}^{N}\partial_{z_k}\chi_{l}(q,p)\right|=C<\infty$.
Let $m>m^p+d$. Using dominated convergence, we get
\begin{align}
&\Ical^\eps\left(\kappa;(x_j-X_j^\kappa(q,p))u;\Theta^x,\Theta^y\right)\phi(x)\nonumber\\
&=\frac{1}{(2\pi\eps)^{3d/2}}\sum\limits_{l=0}^{\infty}\int_{\Rbb^{3d}}\chi_{l}(q,p)(x_j-X^\kappa_{j}(q,p)) e^{\frac{i}{\eps}\Phi^\kappa}(L^\dagger_{y})^m(u\phi)(x,y,q,p)\;dq\, dp\, dy\nonumber\\
&=\frac{1}{(2\pi\eps)^{3d/2}}\sum\limits_{l=0}^{\infty}\int_{\Rbb^{3d}}\chi_{l}(q,p)(x_j-X^\kappa_{j}(q,p)) e^{\frac{i}{\eps}\Phi^\kappa}(u\phi)(x,y,q,p)\;dq\, dp\, dy\nonumber\\
&=\frac{-\eps}{(2\pi\eps)^{3d/2}}\sum\limits_{l=0}^{\infty}\int_{\Rbb^{3d}}\chi_{l}(q,p)e^{\frac{i}{\eps}\Phi^\kappa}\sum_{k=1}^{d}\partial_{z_k}\left(\mathcal{Z}^{-1}_{jk}u\right)(x,q,p)\phi(y)\;dq\, dp\, dy\label{eq:PI_1}\\
&\quad-\frac{\eps}{(2\pi\eps)^{3d/2}}\sum\limits_{l=0}^{\infty}\int_{\Rbb^{3d}}\sum_{k=1}^{d}\left(\partial_{z_k}\chi_{l}\right)(q,p)e^{\frac{i}{\eps}\Phi^\kappa}(\mathcal{Z}^{-1}_{jk}u\phi)(x,y,q,p)\;dq\, dp\, dy\label{eq:PI_2}
\end{align}
After introducing $m$ powers of $L_y^\dagger$, the integrand in~\eqref{eq:PI_2} is dominated by
$$
C\sum_{k=1}^d\left|\mathcal{Z}^{-1}_{jk}(q,p)e^{\frac{i}{\eps}\Phi^\kappa}(L_{y}^\dagger)^{m}(u\phi)(x,y,q,p)\right|.
$$
Thus~\eqref{eq:PI_2} vanishes, whereas~\eqref{eq:PI_1} is the expected quantity.
\end{proof}

By iterative applications, the previous lemma easily extends to any polynomial $x$ dependence. We will only state the result for quadratic polynomials and refer the reader to the proof of Proposition~\ref{prop:comp_PDO}, where a similar problem is treated, for the general case.
\begin{proposition}\label{prop:x_independent_symbol}
Let $P(q,p,x)=\alpha(q,p)+a(q,p)\cdot x+x\cdot A(q,p) x$ be a quadratic polynomial with coefficients in $S[m_P;2d]$ and $u\in S[m_u;2d]$. Then
$$
P\left(q,p,x-X^\kappa(q,p)\right)u(q,p) \quad\sim\quad \sum\limits_{n=0}^2 \eps^n v_n(q,p),
$$
where the $v_n\in S[m_P+m_u;2d]$  are given by
\begin{align*}
v_0(q,p)&=\alpha(q,p)\, u(q,p)\\
\\
v_1(q,p)&=-\div\left(a(q,p)^\dagger\Zcalinv u(q,p)\right)\\
&\quad+\tr(\Zcalinv X^\kappa_z(q,p) A(q,p))\,u(q,p)\\
\\
v_2(q,p)&=\sum\limits_{k=1}^{d}\div\left(\partial_{z_k}[A(q,p)\Zcalinv e_k\,u(q,p)]^\dagger \Zcalinv\right)
\end{align*}
\end{proposition}

\begin{proof}
The statement follows from two successive applications of Lemma~\ref{lemma:part_int}. The quadratic part contributes to $v_1$ and $v_2$:
\begin{align*}
&\quad\left(x-X^\kappa(q,p)\right)\cdot A(q,p) \left(x-X^\kappa(q,p)\right)u(q,p)\\
&\sim -\eps\div\left(u(q,p)(x-X^\kappa(q,p))^\dagger A(q,p) \Zcalinv\right)\\
&=-\eps\tr\left(\partial_z\left[u(q,p)(x-X^\kappa(q,p))^\dagger A(q,p) \Zcalinv\right]\right)\\
&=\eps\tr\left(A(q,p) \Zcalinv X_z^\kappa(q,p)\right)\,u(q,p)\\
&\qquad-\eps\sum\limits_{k=1}^{d} (x-X^\kappa(q,p))^\dagger\left(\partial_{z_k} (A(q,p)\Zcalinv e_k\, u(q,p))\right)\\
&\sim\eps\tr\left(\Zcalinv X_z^\kappa(q,p)A(q,p) \right)\,u(q,p)\\
&\qquad+\eps^2\sum\limits_{k=1}^{d} \div\left[\partial_{z_k}(A(q,p) \Zcalinv e_k\,u(q,p))^\dagger \Zcalinv\right],
\end{align*}
where we used the invariance of the trace under cyclic permutations.
\end{proof}

\subsection{Proofs of the Propositions~\ref{prop:comp_PDO}--\ref{prop:uniqueness}}
We have now established all auxiliary results necessary for proving our central propositions

\begin{proof} (of Proposition~\ref{prop:comp_PDO})
Let $\phi\in\Scal(\Rbb^d,\Cbb)$. The composition of $op^\eps(h)$ with the FIO applied to $\phi$ is
\begin{align*}
&\quad[op^\eps(h) \Ical^\eps(\kappa;u;\Theta^x,\Theta^y)\varphi](x)\\
&=\frac{1}{(2\pi\eps)^{5d/2}}\int_{\Rbb^{5d}}h\left(\frac{x+w}{2},\xi\right)
e^{\frac{i}{\eps}\Psi^\kappa(x,w,y,q,p;\Theta^x,\Theta^y)}u(q,p)\varphi(y)\; dq\, dp\, dy\, dw\, d\xi,
\end{align*}
where
$
\Psi^\kappa(x,w,y,q,p;\Theta^x,\Theta^y):=\xi\cdot(x-w)+\Phi^\kappa(w,y,q,p;\Theta^x,\Theta^y).
$

We introduce the following combinations of positions and momenta, which correspond to creation and annihilation ``variables''
\begin{align*}
a(x,\xi)&:=\left(\Theta^x\right)^{\frac12} x+i\left(\Theta^x\right)^{-\frac12}\xi,\\
\bar{a}(x,\xi)&:=\left(\Theta^x\right)^{\frac12} x-i\left(\Theta^x\right)^{-\frac12}\xi
\end{align*}
and their ``dual operators''
\begin{align*}
\partial_a&:=\frac12\left(\left(\Theta^x\right)^{-\frac12}\partial_x-i\left(\Theta^x\right)^{\frac12}\partial_\xi\right)\\
\partial_{\bar a}&:=\frac12\left(\left(\Theta^x\right)^{-\frac12}\partial_x+i\left(\Theta^x\right)^{\frac12}\partial_\xi\right).
\end{align*}
The Taylor-expansion of the symbol $h$ to order $2N$ around $\kappa(q,p)$ reads
\begin{align*}
h\left(x,\xi\right)&=\sum\limits_{|\alpha+\beta|\leq 2N}
\frac{1}{\alpha!\beta!}\left(\left(\partial_{a}^\alpha\partial_{\bar{a}}^\beta h\right)\circ\kappa\right)(q,p)
\left(a-Z^\kappa(q,p)\right)^\alpha\left(\bar{a}-\bar{Z}^\kappa(q,p)\right)^\beta\\
&\qquad+\sum\limits_{|\alpha+\beta|=2N+1}\left(a-Z^\kappa(q,p)\right)^\alpha\left(\bar{a}-\bar{Z}^\kappa(q,p)\right)^\beta R_{\alpha,\beta}\left(a,\bar{a},q,p\right)\\
&=h_{\rm T}\left(a-Z^\kappa(q,p),\bar{a}-\bar{Z}^\kappa(q,p)\right)+h_{\rm R}\left(a,\bar{a},q,p\right),
\end{align*}
where
\begin{align*}
&\quad R_{\alpha,\beta}\left(a,\bar{a},q,p\right)\\
&=\frac{|\alpha+\beta|}{\alpha!\beta!}\int\limits_0^1 \sigma^{|\alpha+\beta|-1}\left(\partial^\alpha_a\partial^\beta_{\bar{a}} h\right)\textstyle{\left(x+\sigma\left(X^\kappa(q,p)-x\right),\xi+\sigma\left(\Xi^\kappa(q,p)-\xi\right)\right)}d\sigma.
\end{align*}
We split the integral into a part which contains the Taylor polynomial $h_{\rm T}$ and a remainder containing $h_{R}$.
In the first step, we discuss only the part containing $h_{\rm T}$. We have, dropping the argument
$\left(\frac{x+w}{2},\xi\right)$ of $a$ and $\bar a$,
\begin{align*}
\left(\frac12\left(\Theta^x\right)^{\frac12}\partial_\xi-i\left(\Theta^x\right)^{-\frac12}\partial_w\right) \Psi^\kappa&=a-Z^\kappa(q,p).
\end{align*}
Moreover
\begin{align*}
\left(\frac12\left(\Theta^x\right)^{\frac12}\partial_\xi-i\left(\Theta^x\right)^{-\frac12}\partial_w\right)(a-Z^\kappa(q,p))&=0\\
\left(\frac12\left(\Theta^x\right)^{\frac12}\partial_\xi-i\left(\Theta^x\right)^{-\frac12}\partial_w\right)(\bar{a}-\bar{Z}^\kappa(q,p))&=-i\;\id.
\end{align*}
By integration by parts we get
\begin{align}\label{eq:L_PDO_1}
\left(a-Z^\kappa(q,p)\right)^\alpha\left(\bar{a}-\bar{Z}^\kappa(q,p)\right)^\beta v(q,p)
\sim\frac{\eps^{|\alpha|}\beta!}{(\beta-\alpha)!}\left(\bar{a}-\bar{Z}^\kappa(q,p)\right)^{\beta-\alpha} v(q,p),
\end{align}
where we extended the meaning of ``$\sim$'' in an obvious way. We have
\begin{align*}
&\quad\left(\frac12\left(\Theta^x\right)^{\frac12}\partial_\xi+i\left(\Theta^x\right)^{-\frac12}\partial_w+2i\left(\Theta^x\right)^{\frac12}\Zcalinv\partial_z\right)\left(\bar{a}-\bar{Z}^\kappa(q,p)\right)\\
&=-2i\left(\Theta^x\right)^{\frac12}\Zcalinv\partial_z\bar{Z}^\kappa(q,p)
\end{align*}
and hence
\begin{align}
&\quad\left(\bar{a}-\bar{Z}^\kappa(q,p)\right)^{\gamma}v(q,p)\nonumber\\
&\sim-\frac{2\eps}{\#\gamma}\sum_{k|\gamma_k\neq 0}\div\left(e_k^\dagger\left(\Theta^x\right)^{\frac12}\Zcalinv\left(\bar{a}-\bar{Z}^\kappa(q,p)\right)^{\gamma-e_k} v(q,p)\right) \nonumber\\
&=\frac{\eps}{\#\gamma}\sum_{k|\gamma_k\neq 0}\left(\sum_{m=1}^d(\gamma-e_k)_m\left(\bar{a}-\bar{Z}^\kappa(q,p)\right)^{\gamma-e_k-e_m}
(\Lcal_{(e_k,e_m)} v)(q,p)\right.\label{eq:bricks}\\
&\left.\qquad\qquad\qquad+\left(\bar{a}-\bar{Z}^\kappa(q,p)\right)^{\gamma-e_k}(\Lcal_{e_k} v)(q,p)\right) \nonumber,
\end{align}
where the differential operators $\Lcal_{(e_k,e_m)}$ and $\Lcal_{e_k}$ are given by
\begin{align*}
(\Lcal_{(e_k,e_m)}v)(q,p)&:=2e_k^\dagger\left(\Theta^x\right)^{\frac12}\Zcalinv\partial_z\bar{Z}^\kappa e_m v(q,p)\\
(\Lcal_{e_k}v)(q,p)&:=-2\div\left(e_k^\dagger\left(\Theta^x\right)^{\frac12}\Zcalinv v(q,p)\right)
\end{align*}
and $\#\gamma$ denotes the number of non-zero components of $\gamma$.

The symmetrization by the summation over $k$ allows for the iteration of the procedure. We define the three sets
\begin{align*}
\Gamma_1:=\left\{\gamma\in\Nbb^d\;\middle|\;|\gamma|=1\right\},\quad
\Gamma_2:=\Gamma_1\times\Gamma_1,\quad\Gamma:=\Gamma_1\cup\Gamma_2.
\end{align*}
In expression~\eqref{eq:bricks} the sum is taken over all possible reductions of the multi-index $\gamma$ by elements of the ``brick-sets'' $\Gamma_1$ and $\Gamma_2$. After another integration by parts in all terms with $\left(\bar{a}-\bar{Z}^\kappa(q,p)\right)$-dependence, the sum is taken over all possible reductions of $\gamma$ by elements in $\Gamma\times\Gamma$, which may be considered as a two-step path in $\Gamma$, plus the terms which already led to $\gamma=0$ in the first step. So after the removal of all $\left(\bar{a}-\bar{Z}^\kappa(q,p)\right)$-dependence, the sum is taken over all possible paths in the ``brick-set'' $\Gamma$ which reduce $\gamma$ to zero.
To formalise this idea, we define the map 
\begin{align*}
[\;\cdot\;]\quad&:\qquad\Gamma\to\Nbb^d\\
[\gamma]\quad&:=\quad\begin{cases}
\gamma&\gamma\in\Gamma_1\\
\gamma_{1}+\gamma_{2}&\gamma=(\gamma_{1},\gamma_{2})\in\Gamma_2.
\end{cases}
\end{align*}
With
$$
\lambda(\gamma,\gamma_1,\ldots,\gamma_n)=
\begin{cases}
\left(\#(\gamma\!-\!\sum\limits_{l< n}[\gamma_{l}])\right)^{-1}&\gamma_n\in\Gamma_1\\
\left(\#(\gamma\!-\!\sum\limits_{l< n}[\gamma_{l}])\right)^{-1}\!\!\!\left(\gamma-\sum\limits_{l< n}[\gamma_{l}]-e_j\right)_k &\gamma_n=(e_j,e_k)\in\Gamma_2
\end{cases}
$$
we have
\begin{align}
&\left(\bar{a}-\bar{Z}^\kappa(q,p)\right)^{\gamma}v(q,p)\label{eq:L_PDO_2}\\
\sim&\sum_{\substack{
\gamma_1\ldots,\gamma_k\in\Gamma\\
[\gamma_1]+\ldots+[\gamma_k]=\gamma\nonumber
}}\eps^k
\lambda(\gamma,\gamma_1,\ldots,\gamma_k)\ldots\lambda(\gamma,\gamma_1,\gamma_2)\lambda(\gamma,\gamma_1)\left(\Lcal_{\gamma_k}\ldots\Lcal_{\gamma_1}v\right)(q,p).
\end{align}

Combining \eqref{eq:L_PDO_1} and \eqref{eq:L_PDO_2}, we get
\begin{align*}
&\quad(h_T u)(q,p)\\
&\sim\sum_{n=0}^N \eps^nL_n[h;\kappa;\Theta^x,\Theta^y]u(q,p)+\eps^{N+1}\tilde{L}^\eps_{N+1}[h;\kappa;\Theta^x,\Theta^y]u(q,p)\\
&=\sum\limits_{\substack{|\beta|\leq 2N\\\alpha\leq\beta}}\sum_{\substack{
\gamma_1\ldots,\gamma_k\in\Gamma\\
[\gamma_1]+\ldots+[\gamma_k]=\beta-\alpha
}}\hspace*{-5mm}
\frac{\eps^{|\alpha|+k}}{\alpha!(\beta-\alpha)!}\left(\prod_{l=1}^{k}
\lambda(\gamma,\gamma_1,\ldots,\gamma_l)\Lcal_{\gamma_l}\right)\left(u\,\partial_{a}^\alpha\partial_{\bar{a}}^\beta h\circ\kappa\right)(q,p)\\
&\qquad +\eps^{N+1}\tilde{L}^\eps_{N+1}[h;\kappa;\Theta^x,\Theta^y]u(q,p)
\end{align*}
where $\eps^{N+1}\tilde{L}^\eps_{N+1}[h;\kappa;\Theta^x,\Theta^y]u(q,p)$ contains all the terms of order $\eps^{N+1}$ and higher. As $k$ ranges between $\lceil|\beta-\alpha|/2\rceil$ and $|\beta-\alpha|$, we have
\begin{align*}
&\quad L_n[h;\kappa;\Theta^x,\Theta^y]\\
&=\sum_{\substack{n\leq|\alpha+\beta|\leq 2n\\\alpha\leq\beta}}\sum_{\substack{
\gamma_1,\ldots,\gamma_{n-|\alpha|}\in\Gamma\\
[\gamma_1]+\ldots+[\gamma_{n-|\alpha|}]=\beta-\alpha
}}\hspace*{-5mm}
\frac{1}{\alpha!(\beta-\alpha)!}\left(\prod_{l=1}^{n-|\alpha|}
\lambda(\gamma,\gamma_1,\ldots,\gamma_l)\Lcal_{\gamma_l}\right)
\end{align*}
with the following convention for $n-|\alpha|=0$
$$
\sum_{\gamma_1,\ldots,\gamma_{n-|\alpha|}}\prod_{l=1}^{n-|\alpha|}\lambda(\gamma,\gamma_1,\ldots,\gamma_l)\Lcal_{\gamma_l}=\id.
$$
For the first few terms in the expansion, we have more transparent expressions. The zeroth order term
$$
\left(h\circ\kappa\right)(q,p)u(q,p)
$$
is provided by $\alpha=\beta=0$. For the first order term, there are three contributions.
\begin{enumerate}
\item The terms with $|\beta|=1, \alpha=\beta$, which result in
$$
\eps\tr\left(\left(\left(\partial_a\partial_{\bar{a}}h\right)\circ\kappa\right)(q,p)\right)u(q,p).\\
$$
\item The terms $ |\beta|=1, \alpha=0$, which give
$$
-\eps\div\left(\left(\left(\partial_ {\bar{a}}h\right)^\dagger\circ\kappa\right)(q,p)\left(\Theta^x\right)^{\frac12}\Zcalinv u(q,p)\right).
$$
\item The first order contribution of terms $|\beta|=2, \alpha=0$, which is
$$
\eps\tr\left(\Zcalinv\partial_z\bar{Z}^\kappa(q,p)\left(\left({\rm Hess}_{\bar a}h\right)\circ\kappa\right)(q,p)\left(\Theta^x\right)^{\frac12}\right)u(q,p).
$$
\end{enumerate}
By an application of the chain rule, they may be combined to
\begin{align*}
&=-\eps\div\left(\left(\left(\partial_ {\bar{a}}h\right)^\dagger\circ\kappa\right)(q,p)\left(\Theta^x\right)^{\frac12}\Zcalinv u(q,p)\right)\\
&\quad+\eps\tr\left(\Zcalinv \partial_z\left(\left(\Theta^x\right)^{\frac12}\left(\left(\partial_{\bar{a}}h\right)\circ\kappa\right)(q,p)\right)\right)u(q,p).
\end{align*}
The second order term arises in a similar way.

The form of the coefficients of the differential operators $L_n[h;\kappa;\Theta^x,\Theta^y]$ follows, if $\Zcalinv$ is expressed by the formula of minors.
With respect to the symbol class of $v_n$, it is sufficient to note that $\kappa\in S[1;2d]$.

We turn to the discussion of the remainder. The $(a-Z^\kappa(q,p))$ and $(\bar{a}-\bar{Z}^\kappa(q,p))$-factors may be converted to $\eps$-dependence analogously to $h_T$, resulting in terms of order $\eps^{N+1}$ to $\eps^{2N+2}$. As $h(x,\xi)$ is polynomial in $\xi$, the resulting expression equals the  application of a differential operator of order $m_\xi$ to an FIO. We have $\partial_x\Phi^\kappa(x,y,q,p;\Theta^x,\Theta^y)=\Xi^\kappa(q,p)+i\Theta^x(x-X^\kappa(q,p))$, hence the symbol class by iterative applications of Lemma~\ref{lemma:part_int}.
\end{proof}
\begin{remark}
In the case $\kappa=\id, u(q,p)=1, \Theta^x=\Theta^y=\id$, the proof provides the asymptotic expansion of the Anti-Wick symbol of a Weyl quantised pseudodifferential operator.

We have $\Zcal=2\,\id$, $Z^\kappa=q+ip$ and $\bar{Z}^\kappa=q-ip$. Hence $\partial_z\bar{Z}^\kappa=0$, $\Lcal_{(e_j,e_k)}=0$ and $(\Lcal_{e_j}u)(q,p)=-(\partial_{z_j}u)(q,p)$. Moreover, all $\Lcal_{e_j}$ commute.

By straightforward calculation, the Anti-Wick symbol is
\begin{align*}
\sum_{|\beta|=0}^{N}\frac{(-1)^{|\beta|}\eps^{|\beta|}}{\beta!}\left({\rm\Delta}^\beta h\right)(q,p),
\end{align*}
where
$$
{\rm\Delta}^\beta=\prod_{k=1}^d(\partial_{x_k}^2+\partial_{\xi_k}^2)^{\beta_k}.
$$
Thus formally
$$
h_{\rm AW}=e^{-\eps{\rm\Delta}}h_{\rm Weyl},
$$
so we recover that the Anti-Wick quantisation is the solution of the Cauchy-problem for the inverse heat-equation at time $t=\eps$ with the Weyl-symbol as initial datum (compare{\rm ~\cite{[Martinez]}}, where the transition is expressed by the heat-kernel).
\end{remark}

Next, we give the easy proof of Proposition~\ref{prop:comp_time}:
\begin{proof} (of Proposition~\ref{prop:comp_time})
By direct computation
$
i\eps\dt \Ical^\eps\left(\kappa^t;u\right)
$
is an FIO with symbol
\begin{align*}
&\quad i\eps\dt u(t,q,p)-u(t,q,p)
\dt \Phi^{\kappa^t}(x,y,q,p;\Theta^x(t),\Theta^y)\\
&=i\eps\dt u(t,q,p)-u(t,q,p)
\left[\dt S^{\kappa^t}(q,p)-\dt X^{\kappa^t}(q,p)\cdot\Xi^{\kappa^t}(q,p)\right.\\
&\qquad\qquad\qquad+\left(\dt \Xi^{\kappa^t}(q,p)-i\left(\Theta^x(t)\dt X^{\kappa^t}(q,p)\right)\right)
\cdot\left(x-X^{\kappa^t}(q,p)\right)\\
&\qquad\qquad\qquad+\frac{i}{2}\left.\left(x-X^\kt(q,p)\right)\cdot\dt\Theta^x(t)\left(x-X^\kt(q,p)\right)\right]
\end{align*}
The expressions~\eqref{eq:v0_time} and \eqref{eq:v1_time} follow from Proposition~\ref{prop:x_independent_symbol}.
\end{proof}

We close this section with the proof of Proposition~\ref{prop:uniqueness}.
\begin{proof}(of Proposition~\ref{prop:uniqueness})
The proof relies on the inner product
\begin{equation}\label{eq:inner_product}
\left\langle g_{\kappa'(q_0,p_0)}^{\eps,\bar{\Theta^x}},\Ical^\eps\left(\kappa;w;\Theta^x,\Theta^y\right)g_{(q_0,p_0)}^{\eps,\Theta^y}\right\rangle
\end{equation}
for symbols $w\in S[0;2d]$ and canonical transformations $\kappa,\kappa'$ of class $\Bcal$, where
$$
g^{\eps,\Theta}_{(q,p)}(x)=\frac{\det\left(\Re\Theta\right)^\frac14}{(\pi\eps)^{d/4}}e^{-(x-q)\cdot\Theta(x-q)/2\eps}e^{ip\cdot(x-q)/\eps}.
$$

Straightforward calculation gives
$$
\eqref{eq:inner_product}=\frac{2^{d}}{(\pi\eps)^{d}}\frac{\det\left(\Re \Theta^x\right)^\frac14\det\left(\Re\Theta^y\right)^{\frac14}}
{\det\left(2\Theta^x\right)^\frac12\det\left(2\Theta^y\right)^\frac12}
 \int e^{\frac{i}{\eps}\Psi^{\kappa,\kappa'}(q_0,p_0,q,p)}w(q,p)\;dq\,dp,
$$
where
\begin{align*}
&\Psi^{\kappa,\kappa'}(q,p,q_0,p_0)=S^\kappa(q,p)\\
&+\frac12(q-q_0)(p+p_0)-\frac12(X^\kappa(q,p)-X^{\kappa'}(q_0,p_0))(\Xi^\kappa(q,p)+\Xi^{\kappa'}(q_0,p_0))\\
&+i\left(q_0-q\right)\cdot\Theta^y\left(q_0-q\right)/4
+i\left(p_0-p\right)\cdot\left(\Theta^y\right)^{-1}\left(p_0-p\right)/4\\
&+i\left(X^{\kappa'}(q_0,p_0)-X^\kappa(q,p)\right)\cdot\Theta^x\left(X^{\kappa'}(q_0,p_0)-X^\kappa(q,p)\right)/4\\
&+i\left(\Xi^{\kappa'}(q_0,p_0)-\Xi^\kappa(q,p)\right)\cdot\left(\Theta^x\right)^{-1}\left(\Xi^{\kappa'}(q_0,p_0)-\Xi^\kappa(q,p)\right)/4.
\end{align*}
We chose $\sigma\in C_0^\infty(\Rbb^{2d},\Rbb)$ with $\sigma=1$ in an neighborhood of $(q_0,p_0)$ and split the integral into
\begin{align}
&\left\langle g_{\kappa'(q_0,p_0)}^{\eps,\bar{\Theta^x}},\Ical^\eps\left(\kappa;\sigma w;\Theta^x,\Theta^y\right)g_{\kappa(q_0,p_0)}^{\eps,\Theta^y}\right\rangle\label{eq:inner_product1}\\
+&\left\langle g_{\kappa'(q_0,p_0)}^{\eps,\bar{\Theta^x}},\Ical^\eps\left(\kappa;(1-\sigma)w;\Theta^x,\Theta^y\right)g_{\kappa(q_0,p_0)}^{\eps,\Theta^y}\right\rangle\label{eq:inner_product2}.
\end{align}

It is easily seen that
$$
\Im\Psi(q,p,q_0,p_0)=0\textrm{ and }(\nabla_{(q,p)}\Re\Psi)(q,p,q_0,p_0)=0
$$
if and only if  $(q,p)=(q_0,p_0)$ and $\kappa(q_0,p_0)=\kappa'(q_0,p_0)$. Thus the phase in~\eqref{eq:inner_product2} is non-stationary on the support of $w(1-\sigma)$, so after integrations by parts with the operator
$$
\frac{-i\eps}{\|\nabla_{(q,p)}\Psi(q_0,p_0,q,p)\|^2} \bar{\nabla_{(q,p)}\Psi(q_0,p_0,q,p)}\cdot\nabla_{(q,p)}
$$
we have $\lim_{\eps\to 0}\eqref{eq:inner_product2}=0$. By the same argument, the case
$\kappa(q_0,p_0)\neq\kappa'(q_0,p_0)$ gives $\lim_{\eps\to 0}\eqref{eq:inner_product1}=0$.
In the case $\kappa(q_0,p_0)=\kappa'(q_0,p_0)$, we have
$$
{\rm Hess}_{(q_0,p_0)}\Psi^{\kappa,\kappa}=\frac{i}{2}
\begin{pmatrix}
\Theta^y&0\\
0&\left(\Theta^y\right)^{-1}
\end{pmatrix}
+\frac{i}{2}
F^\kappa(q_0,p_0)^\dagger
\begin{pmatrix}
\Theta^x&0\\
0&\left(\Theta^x\right)^{-1}
\end{pmatrix}
F^\kappa(q_0,p_0),
$$
at the stationary points, so by the Stationary Phase Theorem (Theorem~7.7.5. in~\cite{[Hoermander]})
$$
\lim_{\eps\to 0}\left\langle g_{\kappa(q_0,p_0)}^{\eps,\bar{\Theta^x}},\Ical^\eps\left(\kappa;\sigma w;\Theta^x,\Theta^y\right)g_{\kappa(q_0,p_0)}^{\eps,\Theta^y}\right\rangle=C[\kappa;\Theta^x;\Theta^y]w(q_0,p_0),
$$
with the non-vanishing constant 
\begin{align*}
&C[\kappa;\Theta^x;\Theta^y]=2^{2d}
\frac{\det\left(\Re \Theta^x\right)^\frac14\left(\Re\Theta^y\right)^{\frac14}}
{\det\left(\Theta^x\right)^\frac12\det\left(\Theta^y\right)^\frac12}\\
&\qquad\det\left(\begin{pmatrix}
\Theta^y&0\\
0&\left(\Theta^y\right)^{-1}
\end{pmatrix}
+F^\kappa(q_0,p_0)^\dagger
\begin{pmatrix}
\Theta^x&0\\
0&\left(\Theta^x\right)^{-1}
\end{pmatrix}
F^\kappa(q_0,p_0)\right)^{-\frac12}.
\end{align*}

Subsuming this discussion, we have
\begin{align*}
0&=\lim_{\eps\to 0}\left\langle g_{\kappa_1(q_0,p_0)}^{\eps,\Theta^x},\left[\Ical^\eps(\kappa_1;u;\Theta^x,\Theta^y)-\Ical^\eps(\kappa_2;v;\Theta^x,\Theta^y)\right]g_{\kappa_1(q_0,p_0)}^{\eps,\Theta^y}\right\rangle\\
&=\begin{cases}
C[\kappa_1;\Theta^x;\Theta^y]u(q_0,p_0)&\kappa_1(q_0,p_0)\neq\kappa_2(q_0,p_0)\\
C[\kappa_1;\Theta^x;\Theta^y]u(q_0,p_0)-C[\kappa_2;\Theta^x;\Theta^y]v(q_0,p_0)&\kappa_1(q_0,p_0)=\kappa_2(q_0,p_0)
\end{cases}
\end{align*}
In the case $\kappa_1(q_0,p_0)\neq\kappa_2(q_0,p_0)$ we immediately get $u(q_0,p_0)=0$ and by symmetry $v(q_0,p_0)=0=u(q_0,p_0)$.
In the case $\kappa_1(q_0,p_0)=\kappa_2(q_0,p_0)$, we either have $u(q_0,p_0)=0$ or $u(q_0,p_0)\neq 0$. In the first case, we immediately get $v(q_0,p_0)=0=u(q_0,p_0)$. In the second case $u$ does not vanish in a neighbourhood of $(q_0,p_0)$. Hence $\kappa_1=\kappa_2$ in the same neighborhood, thus $C[\kappa_1;\Theta^x;\Theta^y]=C[\kappa_2;\Theta^x;\Theta^y]$ and so $u(q_0,p_0)=v(q_0,p_0)$.
\end{proof}

\section{Uniform approximations of the unitary group}
In this section we will combine the results of Section~\ref{section:composition} to our main result.
\begin{theorem}\label{theo:main_result}
Let $U^\eps(t,s)$ be the propagator associated to the time-dependent Schr\"odinger-equation
$$
i\eps\dt \psi^\eps(t)=H^\eps(t)\psi^\eps(t),\qquad\psi^\eps(s)=\psi^\eps_s\in L^2(\Rbb^d,\Cbb)
$$
on the time-interval $-T\leq s,t\leq T$ where $H^\eps(t)=op^\eps(h_0+\eps h_1)$ with subquadratic $h_0(t,x,\xi)$ and sublinear $h_1(t,x,\xi)$, both polynomial in $\xi$. Moreover let
\begin{itemize}
\item $\Theta^y\in\mathrm{Gl}(d)$ be complex symmetric with positive definite real part and
\item $\Theta^x\in C^1(\Rbb,\mathrm{Gl}(d))$ be complex symmetric  fulfilling $0<\gamma\,\id\leq\Re\Theta^x(t)\leq\gamma'\,\id$ for all $t\in[-T,T]$ in the sense of quadratic forms.
\end{itemize}
Then
$$
\sup_{-T\leq s,t\leq T}\left\|U^\eps(t,s)-\Ical^\eps\left(\kt;\sum\limits_{n=0}^N\eps^n u_n;\Theta^x(t),\Theta^y\right)\right\|_{L^2\to L^2}\!\!\!\!\!\!\!\!\leq C(T)\eps^{N+1},
$$
where $\kt$ and the $u_n$ are uniquely given as
\begin{itemize}
\item the Hamiltonian flow associated to $h_0$ and
\item the solutions of
\begin{align*}
&\dt u_n\left(t,s,q,p\right)=\frac{1}{2}u_{n}(t,s,q,p) \times\\
&\quad\left[\tr\left(\Zcalinvts\dt \Zcalts\right)-i h_1\left(t,X^\kt(q,p),\Xi^\kt(q,p)\right)\nonumber
\vphantom{\tr\left(Z^{-1}(F^\kt(q,p);\Theta^x,\Theta^y)\dt Z(F^\kt(q,p);\Theta)\right)}\right]\nonumber\\
&-\sum\limits_{k=1}^{d}\div\left[\left(\partial_{z_k}\left(\dt\Theta^x(t)\Zcalinvts e_k\,u_{n-1}\right)\right)^\dagger
\Zcalinvts\right]\\
& +i\sum_{j=2}^{n-2}L_{j}[h_0(t);\kt;\Theta^x(t),\Theta^y] u_{n+1-j}\\
& +i\sum_{j=1}^{n-3}L_{j}[h_1(t);\kt;\Theta^x(t),\Theta^y] u_{n-j}\nonumber
\end{align*}
with initial conditions 
\begin{align*}
u_0(s,s,q,p)&=\det\left(\Theta^x(s)+\Theta^y\right)^{1/2}\\
u_n(s,s,q,p)&=0,\quad n\geq 1.
\end{align*}
\end{itemize}
\end{theorem}

\begin{corollary}\label{coro:ehrenfest}
Under the additional assumption
$$
\sum_{2\leq |\alpha|\leq 4d+7}\|(\partial^\alpha_{(q,p)} h)(t,q,p)\|_{L^\infty(\Rbb\times\Rbb^{2d})}<\infty,
$$
we have the following result on the Ehrenfest timescale $T(\eps)=C_T\log(\eps^{-1})${\rm :}
$$
\sup_{-T(\eps)\leq s,t\leq T(\eps)}\left\|U^\eps(t,s)-\Ical^\eps\left(\kt;u_0;\Theta^x(t),\Theta^y\right)\right\|\leq C(C_T)\eps^{1-\rho(C_T)},
$$
where $\rho(C_T)$ can be made arbitrary small, if $C_T$ is chosen small enough.
\end{corollary}
\begin{remark}\hfill
\begin{enumerate}
\item We recall that $\Zcalts=(i\left(\Theta^y\right)^{-1}\:\:\id)F^\kt(q,p)^\dagger(-i\Theta^x(t)\:\:\id)^\dagger$,
thus the dependence of $\Zcalts$ on $q$ and $p$ is only via the Jacobian of the flow.

\item
The expression for the leading order symbol is
\begin{align*}
u_0(t,s,q,p)&=\left(\det\Theta^y\Zcalts\right)^\frac12 e^{-i\int\limits_s^t h_1\left(\tau,\kappa^{(\tau,s)}(q,p)\right)d\tau},
\end{align*}
where the branch of the square root is defined by continuity in $t$ starting from $u_0(s,s,q,p)=\det\left(\Theta^x(s)+\Theta^y\right)^\frac12$, compare~Remark~\ref{remark:identity}. The corresponding FIO is known as the Herman-Kluk propagator in the chemical literature. Notice that the dependence of $u$ on $q$ and $p$ is only via $F^\kt(q,p)$. Likewise, the $(q,p)$-dependence in $u_k$ is only via $F^\kt(q,p)$ and its derivatives with respect to $q$ and $p$.

\item As an easy corollary we get that the FIO defined in the last theorem approximately inherits the properties of $U(t,s)$, i.e. it is almost unitary in the sense that
\begin{align*}
\left\|\Ical^\eps\left(\kt;u;\Theta^x(t),\Theta^y\right)\Ical^\eps\left(\kt;u;\Theta^x(t),\Theta^y\right)^*-\id\right\|\leq& C_N\eps^{N+1}\\
\left\|\Ical^\eps\left(\kt;u;\Theta^x(t),\Theta^y\right)^*\Ical^\eps\left(\kt;u;\Theta^x(t),\Theta^y\right)-\id\right\|\leq& C_N\eps^{N+1},
\end{align*}
where
$u=\sum_{n=0}^N\eps^n u_n$ and it almost fulfills the group property, i.e.
$$
\left\|\Ical^\eps\left(\kappa^{(t,t')}\!;u;\id,\id\right)\Ical^\eps\left(\kappa^{(t',s)}\!;u;\id,\id\right)\!-\!\Ical^\eps\left(\kt\!;u;\id,\id\right)\right\|\leq C_N'\eps^{N+1}.
$$
The result also holds for general $\Theta^x$ and $\Theta^y$. The possibility of stating the correct dependence of the symbol on the matrices is left to the reader.
\item
In the case of a linear flow
$$
\kappa^{(t,s)}(q,p)=M(t,s)
\begin{pmatrix}q\\p\end{pmatrix}
$$
the approximation becomes exact as the symbols $v_n$ in Propositions~\ref{prop:comp_PDO} and~\ref{prop:comp_time} vanish for $n\geq 2$. This can be rephrased as
$$
\Ical^\eps(\kt;u_0;\Theta^x,\Theta^y)=\mathrm{Met}^\eps\left( F^\kt(q,p)\right),
$$
where $\mathrm{Met}^\eps$ denotes the $\eps$-dependent metaplectic representation.

\item The Ehrenfest result does not generalise to the higher order corrections. This is in accordance with the situation in~{\rm \cite{[Bambusi]}}, where holomorphicity of the symbol is required to get higher order versions of Egorov's Theorem on the Ehrenfest timescale.

\item The proof will produce the following byproduct: the so-called Frozen Gaussian Approximation, which is obtained by still choosing $\kt$ as the Hamiltonian flow but keeping $u_0$ constant for all $q,p$ and $t$, is an asymptotic solution of the Schr\"odinger equation of order $O(\eps)$. Thus it approximates the unitary group for the short times of order $\eps$. It will not be a valid approximation for longer times because of the uniqueness of the symbol.
\end{enumerate}
\end{remark}
\begin{proof}
By Theorem~\ref{theo:FIO}, an FIO associated to a $C^1$ family $\kt$ of canonical transformation of class $\Bcal$ and $(x,y)$-independent symbol $u=\sum_{n=0}^{N}\eps^n u_n$, $u_n\in C^1(\Rbb,S[0;2d])$ leaves $\Scal(\Rbb^d,\Cbb)$ invariant. Thus, we can plug such an operator as an ansatz into the time-dependent Schr\"odinger equation~\eqref{eq:TDSE}. By Propositions~\ref{prop:comp_PDO} and \ref{prop:comp_time} we have a representation
\begin{align*}
&\quad\left(i\eps\dt -op^\eps(h_0+\eps h_1)\right)\Ical^\eps\left(\kt;\sum_{n=0}^{N}\eps^n u_n;\Theta^x(t),\Theta^y\right)\\
&=\Ical^\eps\left(\kt;\sum_{n=0}^{N+1}\eps^n v_{n};\Theta^x(t),\Theta^y\right)
+\eps^{N+2} \Ical^\eps\left(\kt;v^\eps_{N+2};\Theta^x(t),\Theta^y\right)
\end{align*}
on $\Scal(\Rbb^d,\Cbb)$. We will show that the $v_n$, $0\leq n\leq N+1$ vanish, if $\kt$ and $u_n\in S[0,2d]$, $0\leq n\leq N$ are chosen properly. Moreover, it will turn out that $v^\eps_{N+2}$ is of class $S[0;3d]$. Thus, by Theorem~\ref{theo:FIO}, $\Ical^\eps\left(\kt;\sum_{n=0}^{N}\eps^n u_n;\Theta^x(t),\Theta^y\right)$ is an asymptotic solution of order $N+2$. The statement will then follow by Lemma~\ref{lemma:gronwall}.

Combining~Propositions~\ref{prop:comp_PDO} and~\ref{prop:comp_time}, one recognises
$v_0$ as the product of $u_0$ and
\begin{align}
\left(-\dt S^\kt(q,p)+\dt X^\kt(q,p)\cdot\Xi^\kt(q,p)-h_0\left(t,\kt(q,p)\right)\right)\label{eq:v0_HK}.
\end{align}
As we do not expect $\Ical^\eps(\kappa;0;\Theta^x,\Theta^y)=0$ to be a good approximation of $U(t,s)$, we require~\eqref{eq:v0_HK} to vanish. By combining its derivatives with respect to $p$ and $q$, it is easily seen that this is the case if and only if $\kt$ is the Hamiltonian flow associated to $h_0$.

There are several parts which contribute to $v_1$:
\begin{enumerate}
\item the zeroth order terms of Propositions~\ref{prop:comp_PDO} and~\ref{prop:comp_time} applied to $u_1$,
\item the first order terms of Propositions~\ref{prop:comp_PDO} and~\ref{prop:comp_time} applied to $u_0$,
\item the zeroth order term of Proposition~\ref{prop:comp_PDO} applied to $u_0$ for the subprincipal symbol $h_1$.
\end{enumerate}
Thus we get the following expression for $v_1$:
\begin{align}
&\quad u_1\left[-\dt S^\kt\!(q,p)+\dt X^\kt\!(q,p)\cdot\Xi^\kt(q,p)-h_0\left(t,\kt(q,p)\right)\right]\label{eq:v1_1}\\
&\quad +\div\left(\left[(\partial_x h_0+i\Theta^x(t)\partial_\xi h_0)\left(t,\kt(q,p)\right)\right]^\dagger \Zcalinvts u_0\right)\label{eq:v1_2}\\
&\quad +\div\left(\left(\dt \Xi^\kt(q,p)-i\Theta^x(t)\dt X^\kt(q,p)\right)^\dagger \Zcalinvts u_0\right)\label{eq:v1_3}\\
&\quad -u_0\frac{1}{2}\tr\left(\Zcalinvts \partial_{z}\left[(\partial_x h_0+i\Theta^x(t)\partial_\xi h_0)\left(t,\kt(q,p)\right)\right]\right) \nonumber\\
&\quad +i\dt u_0-\frac{i}{2}u_0\tr\left(\Zcalinvts X^\kappa_z(q,p)\dt\Theta^x(t)\right)\nonumber\\
&\quad -u_0h_1\left(t,\kt(q,p)\right)\nonumber\\
&=i\dt u_0-\frac{i}{2}u_0\tr\left(\Zcalinvts \dt \Zcalts\right)-u_0h_1\left(t,\kt(q,p)\right)\nonumber,
\end{align}
as~\eqref{eq:v1_1}, \eqref{eq:v1_2} and \eqref{eq:v1_3} vanish because of the choice of $\kt$ as the Hamilton flow.

As the linearisation of $\det(A)$ is $\det(A)\tr(A^{-1} \mathrm{d}A)$ for invertible $A$, the equation $v_1=0$ with initial conditions that recover identity, is solved by
$$
u_0(t,s,q,p)=\left(\det(\Theta^y\Zcalts)\right)^{\frac{1}{2}} \exp\left(-i\int\limits_s^t h_1\left(\tau,\kappa^{(\tau,s)}(q,p)\right)d\tau\right),
$$
which is of class $S[0;2d]$.

$v_n$ is given by the following expression:
\begin{align*}
&i\dt u_{n-1}-\frac{i}{2}u_{n-1}\tr\left(\Zcalinvts \dt \Zcalts\right)
-u_{n-1}h_1\left(t,\kt(q,p)\right)\\
&-\sum\limits_{k=1}^{d}\div\left(\left(\partial_{z_k}\left(\dt\Theta^x(t)\Zcalinvts e_k\,u_{n-2}\right)\right)^\dagger \Zcalinvts\right)\\
& -i\sum_{j=2}^{n-2}L_{j}[h_0(t);\kt;\Theta^x(t),\Theta^y] u_{n-j}\\
& -i\sum_{j=1}^{n-3}L_{j}[h_1(t);\kt;\Theta^x(t),\Theta^y] u_{n-j-1}\nonumber,
\end{align*}
where we already dropped the terms analogous to~\eqref{eq:v1_1}--\eqref{eq:v1_3}. The equation $v_n=0$ is easily solved by variation of the constant and has a solution in $S[0;2d]$.

We finally note that the highest order symbol is of class $S[0;3d]$. Thus, we have established that the constructed FIO is an asymptotic solution of order $N+2$ on the class of Schwartz functions, so the result follows by the strategy outlined at the beginning of the proof.

We turn to the uniqueness. Assume that there are $\tilde{\kappa}^{(t,s)}$ and $\tilde{u}\in S[0;2d]$ such that 
$$
\left\|U^\eps(t,s)-\Ical^\eps\left(\tilde{\kappa}^{(t,s)};\tilde{u};\Theta^x(t),\Theta^y\right)\right\|\leq C'(T)\eps.
$$
In this case we have
$$
\left\|\Ical^\eps\left(\kappa^{(t,s)};u_0;\Theta^x(t),\Theta^y\right)-\Ical^\eps\left(\tilde{\kappa}^{(t,s)};\tilde{u};\Theta^x(t),\Theta^y\right)\right\|\leq \left(C(T)+C'(T)\right)\eps
$$
and thus we get $\tilde{\kappa}^{(t,s)}=\kt$ and $\tilde{u}=u_0$ on $\supp u_0=\Rbb^{2d}$ by Proposition~\ref{prop:uniqueness}. The uniqueness of the higher order corrections follows inductively by the same kind of argument.
\end{proof}

\begin{proof}(of Corollary~\ref{coro:ehrenfest})
To extend the result to the Ehrenfest timescale, we have to study the dependence of the remainder's symbol $v_2^\eps$ on $T=T(\eps)$ and to show that the growth of
$$
\sum_{|\alpha|\leq 4d+1}\left\|\partial^\alpha_x v_2^\eps\right\|_{L^\infty}
$$
does not exceed $O(\eps)$ in the $\eps\to 0$ limit if $C_T$ is sufficiently small.

The dependence comes from the elements of $F^\kt(q,p)$ and its derivatives. By ~\eqref{eq:bound_ehrenfest}, they allow for a bound of the form $C'(C_T) e^{-\rho(C_T)}$, where $\rho(C_T)$ can be made arbitrary small if $C_T$ is chosen small enough. Moreover, $v^\eps_2$ has polynomial growth in these quantities, which follows from the form of the differential operators of Proposition~\ref{prop:comp_PDO}, the explicit expression of $u_0$ and the bound away from zero of the determinant of $\Zcalts$, compare the proof of Lemma~\ref{lemma:Z_inv}. Combining these facts, the result follows.
\end{proof}

\begin{appendix}
\section{Gaussian integrals with non-real matrices}
We consider the convex cone $\Ccal$ of complex symmetric matrices with positive definite real part. Every matrix of $\Ccal$ is invertible with its spectrum included in the open half plane $\{z|\Re z>0\}$. It follows from matrix theory (see~\cite{[JohnsonOkuboReams]}) that each element of $\Ccal$ admits an unique square root in $\Ccal$. Furthermore, the square root of $M$ is given by the Dunford-Taylor integral (see~\cite{[Kato]} I.\textsection 5.6)
\begin{equation*}
M^{1/2}=\frac{1}{2\pi i}\int_\Gamma z^{1/2}(M-z)^{-1}\;d z
\end{equation*}
where the integration path is a closed contour in the half-plane $\{z|\Re z>0\}$ making a turn around each eigenvalue in the positive direction and the value of $z^{1/2}$ is chosen so that it is positive for real positive $z$. As a consequence, the square root $M^{1/2}$ is an holomorphic function of $M$. If one considers the computation of the Gaussian integral
\begin{equation*}
\frac{1}{(2\pi\eps)^{d/2}}\int_{\Rbb^d}e^{-\frac{M}{2\eps}x\cdot x}\;dx,
\end{equation*}
it is well-known that its value is given by $(\det M)^{1/2}=\det(M^{1/2})$ for positive definite real symmetric $M$. From the above discussion, it directly follows that this property extends to any matrix $M\in\Ccal$ (see Appendix A in~\cite{[Folland]} or Section 3.4. in~\cite{[Hoermander]} for an alternative explanation).
\end{appendix}

\end{document}